\documentclass[11pt]{article}
\usepackage[top=23.4mm, bottom=23.4mm, left=23.4mm, right=23.4mm]{geometry}
\usepackage{graphicx,dcolumn,bm}
\usepackage{amssymb,amstext,amsmath}

\usepackage{graphicx}
\usepackage{dcolumn}
\usepackage{bm}
\usepackage{amssymb}
\usepackage{multirow}
\usepackage{bigstrut}
\usepackage{makecell,rotating}
\usepackage{mathrsfs}
\usepackage{booktabs}
\usepackage{threeparttable}
\usepackage{multirow}
\usepackage{subfigure}
\usepackage{colortbl}
\usepackage{epsfig}
\usepackage{threeparttable}
\usepackage{chngpage}
\usepackage{float}
\usepackage{amstext}
\usepackage{amsmath}
\usepackage{amsthm}
\usepackage{pdflscape}
\usepackage{authblk}
\usepackage{hyperref}
\usepackage{rotating}
\usepackage[graphicx]{realboxes}
\usepackage{adjustbox}
\usepackage{setspace}
\usepackage{caption}
\usepackage{xcolor}
\usepackage{xr}
\usepackage{caption}
\usepackage{lineno}

\newtheorem{lemma}{Lemma}
\newtheorem{corollary}{Corollary}
\newtheorem{theorem}{Theorem}
\newtheorem{remark}{Remark}
\usepackage{enumitem}
\setitemize[1]{itemsep=1pt,partopsep=0pt,parsep=\parskip,topsep=4pt}
\title{Evolutionary dynamics in repeated optional games}
\author{Fang Chen$^{1\dagger}$, Lei Zhou$^{2\dagger}$ and Long Wang$^{1*}$}
\date{
    $^1$ \footnotesize Center for Systems and Control, College of Engineering, Peking University, Beijing 100871, China\\%
    $^2$ School of Automation, Beijing Institute of Technology, Beijing 100081, China\\
    $^{\dagger}$ These authors contributed equally to this work\\
    $^*$ Corresponding author. E-mail: longwang@pku.edu.cn\\[2ex]
}
\begin{document}
\maketitle

\begin{abstract}
    Direct reciprocity facilitates the evolution of cooperation when individuals interact repeatedly. Most previous studies on direct reciprocity implicitly assume compulsory interactions. Yet, interactions are often voluntary in human societies. Here, we consider repeated optional games, where individuals can freely opt out of each interaction and rejoin later. We find that voluntary participation greatly promotes cooperation in repeated interactions, even in harsh situations where repeated compulsory games and one-shot optional games yield low cooperation rates. Moreover, we theoretically characterize all Nash equilibria that support cooperation among reactive strategies, and identify three novel classes of strategies that are error-robust, readily become equilibria, and dominate in the evolutionary dynamics. The success of these strategies hinges on the effect of opt-out: it not only avoids trapping in mutual defection but also poses additional threats to intentional defectors. Our work highlights that voluntary participation is a simple and effective mechanism to enhance cooperation in repeated interactions.
\end{abstract}

\section*{Introduction}
Humans routinely face social dilemmas where mutual cooperation is most beneficial for the group yet each group member profits more by defecting~\cite{Dawes1980}.
Classical metaphors to describe such social dilemmas include the prisoner's dilemma and the public goods game (PGG)~\cite{Groves1977, Ostrom1990, Perc2013}.
Without additional mechanisms, natural selection generally favors defection in such games, which contrasts with the reality that cooperation is ubiquitous~\cite{Nowak2006, Sigmund2010book}.
This raises a fundamental question about how cooperation evolves~\cite{Axelrod1981,axelrod1984book}.
Based on repeated interactions, one mechanism that is shown to support the evolution of cooperation is direct reciprocity~\cite{Trivers1971, Axelrod1981}, under which individuals cooperate conditionally on past interactions.
Mathematically, the logic of direct reciprocity can be conveniently described by the framework of repeated games.
Indeed, employing this framework, previous work has addressed important questions such as which strategies support cooperation and under what conditions, cooperation evolves~\cite{Friedman1973, Boyd1987,van2012direct, akccay2018collapse, Park2022, Wu2018aspiration,Donahue2020a,Reiter2018a,Chen2022geometry,tan2021payoff}.

A tacit assumption in most previous studies on direct reciprocity is that interactions are compulsory, namely, each individual should participate in every interaction. In reality, participation is often voluntary and individuals have the freedom to opt out~\cite{Earle1987, Orbell1993, Batali1995, Hauert2002}.
In this case, the underlying strategic interactions are better captured by repeated optional games, where individuals are allowed to abstain from any interaction (and also to resume participation) (see Fig.~\ref{fig:model}). When individuals choose to opt out, they become self-sufficient and obtain a payoff that is independent of others. This payoff is often set to be greater than the one received under the social trap of mutual defection and less than the social optimum with everyone cooperating~\cite{Hauert2002, Hauert2007, rand2011evolution}, encouraging individuals to opt out when mutual defection occurs and to re-establish cooperation if individuals abstain. In repeated optional games, opt-out can serve as an additional response against co-players' defection, which guarantees a safe income and avoids the risk of mutual retaliation. Based on these, opting out conditionally on past behaviors may become new leverage to force cooperation in repeated optional games.

Nonetheless, existing studies on repeated optional games fail to provide a comprehensive understanding of the role that opt-out plays in the evolution of cooperation due to (i) a presupposition of a small and incomplete set of available strategies \cite{Batali1995} and (ii) no focus on cooperation \cite{Yamamoto2019, Ahmed2020}.
So far, it is yet to be known which strategies facilitate the evolution of cooperation in repeated optional games if all possible strategies of a given complexity are considered and under what conditions, these strategies dominate. More importantly, it still remains unclear how individuals could strategically opt out to promote cooperation.
Although an interesting finding in one-shot (non-repeated) optional PGGs shows that unconditional opt-out (i.e., always opt out) can rescue cooperation if the incentive for cooperation is high~\cite{Hauert2002, Hauert2002a, Semmann2003}, unconditional strategies are easily invaded and cooperation in one-shot optional games is not stable.

Here, we systematically investigate the effect of voluntary participation on the evolution of cooperation in repeated games. For a comprehensive analysis, we conduct an exhaustive search for optimal strategies in the space of reactive strategies. Through evolutionary simulations, we show that voluntary participation leads to almost full cooperation even in situations where repeated compulsory games and one-shot optional games yield low propensities for cooperation. Resorting to equilibrium analysis, we mathematically characterize all Nash equilibria that support cooperation, and identify three novel classes of strategies that are robust to implementation errors, readily become equilibria, and dominate in the evolutionary dynamics. In the meanwhile, we find that these strategies and their behaviorally close variants account for the evolutionary advantage under voluntary participation and are thus key to the promotion of cooperation in repeated optional games.
For the success of these strategies, their effective leverage of opting out against defection is crucial: it offers a safe income that cannot be exploited, provides a way out of mutual defection, and poses additional threats to intentional defectors.
In addition, when considering the effect of opt-out payoff on cooperation, our results demonstrate that a small incentive for opt-out is enough to achieve almost full cooperation. Besides, all our findings are verified to be robust to changes in model parameters and to other model extensions (e.g., failures of opting out). Our work thus highlights that voluntary participation is a simple and effective mechanism to enhance cooperation in repeated interactions.

\begin{figure}[!ht]
    \centering
    \includegraphics[width=0.8\textwidth]{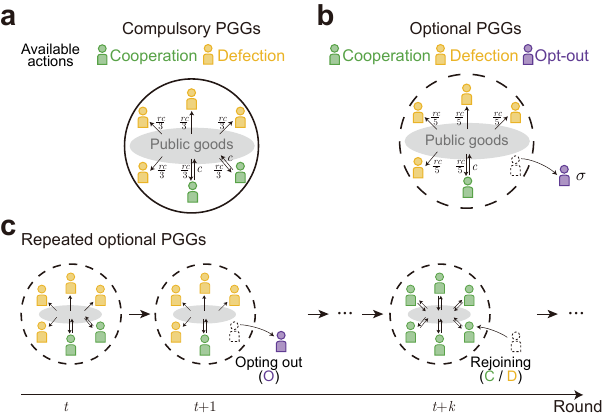}
    \caption{\textbf{In repeated optional public goods games (PGGs), individuals can additionally choose to opt out of the interaction and rejoin later.} \textbf{a}, In a compulsory PGG, each individual has to decide either to cooperate (marked as green) by contributing an amount, $c$, to the public goods or to defect (marked as yellow) by contributing nothing. The total contributions are then multiplied by a factor, $r$, and equally divided among all participants, irrespective of whether they cooperate or defect. \textbf{b}, In an optional PGG, individuals have the additional option to opt out (marked as purple) and gain a payoff $\sigma$ ($0<\sigma<(r-1)c$) that does not depend on others' actions. If there are $x$ individuals cooperating, $y$ defecting, and the number of participants $x+y>1$, an individual who cooperates gets $xrc/(x+y)-c$ and an individual who defects gets $xrc/(x+y)$. If only one individual participates in the game (i.e., $x+y=1$), the interaction is canceled and all individuals get $\sigma$. \textbf{c}, In repeated optional PGGs, individuals interact for many rounds of optional PGGs. In each round, individuals decide to cooperate, defect or opt out depending on the outcome of the previous round. Compared with repeated compulsory PGGs where individuals are required to participate in every interaction, repeated optional PGGs allow individuals to opt out of any interaction and to rejoin later.}
    \label{fig:model}
\end{figure}

\section*{Results}
\paragraph{Repeated optional games.}
In the following, we introduce the framework of repeated optional games. Here, we focus on repeated optional PGGs (see illustrations in Fig.~\ref{fig:model} and see repeated optional prisoner's dilemma games in Section 4 of the Supplementary Information). In such games, there are $n\ge 2$ individuals and they repeatedly play many rounds of optional PGGs. In every round, each individual can choose one of the three actions, to participate in the game and cooperate ($C$) by contributing an amount $c>0$ to the public goods, to participate in the game and defect ($D$) by contributing nothing, and to opt out ($O$) and obtain a fixed payoff $\sigma$. The total contributions in the public goods are then multiplied by a multiplication factor $r$ ($1<r<n$) and evenly distributed to all the participants, irrespective of whether they cooperate or defect. If there are $x$ individuals cooperating, $y$ defecting, and at least two individuals participating in the game (i.e., $2\le x+y\le n$), the payoff for an individual who cooperates, defects, and opts out is $P_C(x,y)=xrc/(x+y)-c$, $P_D(x,y)=xrc /(x+y)$, and $P_O(x,y)=\sigma$, respectively. If less than two individuals choose to participate in the game (i.e., $x+y<2$), the interaction is canceled and everyone obtains the payoff $\sigma$. Here, we assume that $0<\sigma<(n-1)c$, meaning that full cooperation is better off than full opt-out, and full opt-out is better off than full defection \cite{Batali1995, Hauert2002, Vanberg1992}.

We consider repeated optional PGGs that last for infinitely many rounds in the limit of no discounting (see discounted games in the Supplementary Information). In such games, individuals may take the whole game history into account to make a decision, and the resulting strategy can be arbitrarily complex. To make evolutionary analysis feasible, we focus on reactive strategies where current actions depend on the number of each action in the previous round~\cite{Pinheiro2014}. Let $(x,y)$ denote the game state in the previous round, where $x$ and $y$ are respectively the numbers of individuals who cooperate and defect. Let $\mathcal{G}=\{(x,y)|0\leq x,y\leq n \text{ and }0\leq x+y \leq n\}$ denote the set of all possible game states. A reactive strategy can be represented as
\begin{equation}
  \mathbf{p}=(p^C_{x,y},p^D_{x,y})_{(x,y)\in\mathcal{G}},
  \label{equ:reactiveStr}
\end{equation}
where $p^A_{x,y}\in[0,1]$ is the probability to implement action $A$ ($A\in\{C, D\}$) in the current round and $p^C_{x,y} + p^D_{x,y}\le 1$ for all $(x,y)\in\mathcal{G}$. The probability to opt out is thus $p^O_{x,y}=1-p^C_{x,y}-p^D_{x,y}$. A strategy is pure if all entries $p^A_{x,y}$ belong to the set $\{0,1\}$; otherwise, it is stochastic. We also consider the effect of trembling hands (i.e., implementation errors), where individuals may mistakenly implement another random (not intended) action with a small probability $\varepsilon/2$ where $\varepsilon\in[0,1]$. For instance, if a pure strategy $\mathbf{p}$ prescribes cooperation after full cooperation, individuals adopting this strategy may mistakenly defect or opt out of the game with probability $\varepsilon/2$, and correctly cooperate with the rest probability $1-\varepsilon$. Therefore, when errors are present ($\varepsilon>0$), the effective strategy for a pure strategy becomes stochastic and each entry $p^A_{x,y}\in\{\varepsilon/2, 1-\varepsilon\}$.

When individuals adopt strategies $\mathbf{p}_1,\mathbf{p}_2,\ldots, \mathbf{p}_n$ to play repeated optional PGGs, the game dynamics can be modeled as a Markov chain. The state space of the Markov chain is the set of all action profiles. When the effect of trembling hands is considered ($\varepsilon >0$), the Markov chain is ergodic and there exists a stationary distribution. In repeated optional PGGs, the payoff for each strategy is the average gain in this stationary distribution (see Supplementary Information for details).
\bigskip

\paragraph{Evolutionary dynamics.}
On a longer time scale, we assume that individuals change their strategies. Here, we consider the pairwise comparison process~\cite{Szabo1998,Traulsen2006,Zhou2018coevolution,Zhou2021Aspiration,Su2019,Wang2022jointliability}, where individuals imitate successful peers, in a well-mixed population of $N$ individuals. At each time step of such a process, a group of $n$ individuals is randomly selected from the population. They play repeated optional PGGs and each obtains the expected payoff $\pi$. After that, a random individual $l$ is selected to update its strategy. It either adopts a random strategy with probability $\mu$ (random exploration or mutation) or implements imitation with the rest probability $1-\mu$. If individual $l$ imitates, it randomly chooses a role model $k$ ($k\neq l$), and adopts its strategy with a probability that depends on the payoff difference between individual $k$ and $l$, i.e., $\pi^k-\pi^l$. The larger the payoff difference is, the more likely $k$ is imitated (see Supplementary Information for details). When the mutation is present ($\mu>0$), the resulting evolutionary dynamics are ergodic and it is possible to transit between any possible strategy configurations of the population. In this work, we mainly focus on the case of rare mutations ($\mu\rightarrow 0$)~\cite{Fudenberg2006}, where the evolutionary dynamics spend most of the time in homogenous populations with everyone adopting the same strategy.

\begin{figure}[!ht]
    \centering
    \includegraphics[width=0.8\textwidth]{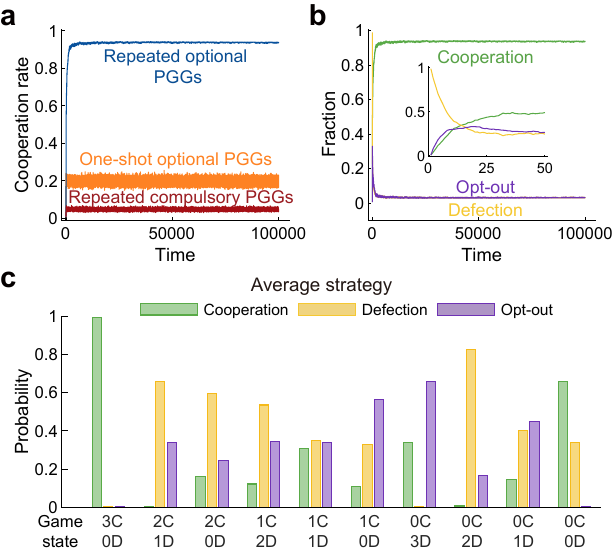}
    \caption{\textbf{Repeated optional PGGs greatly promote cooperation when repeated compulsory PGGs and one-shot optional PGGs yield low propensities for cooperation.} \textbf{a}, We analyze the evolutionary dynamics for repeated three-player optional PGGs ($n=3$). The x-axis represents the number of mutations introduced to the population (see Supplementary Information). The y-axis corresponds to the cooperation rate averaged over $1000$ independent simulations. Strategies of all individuals in the population are initialized as always defecting ($ALLD$). Different initial conditions do not alter the result (see Supplementary Fig. 1). \textbf{b}, Further analysis of repeated optional PGGs shows that the additional option, opt-out, acts as a catalyst for the emergence of cooperation: when the population is in a defective state, opt-out provides a way to avoid the deadlock of mutual defection; the cooperation rate increases as the rate of opt-out grows; eventually, the cooperation rate goes up to almost one. \textbf{c}, We calculate the average strategy that individuals use during the evolution. Each bar represents the probability of cooperation (green bar), defection (yellow bar) and opt-out (purple bar), given that there are $x$ individuals cooperating and $y$ defecting (i.e. game state $(xC,yD)$). This strategy indicates that four game states deserve attention and they are $(3C, 0D)$ (all cooperate), $(2C, 1D)$ (one defects and others cooperate), $(0C, 3D)$ (all defect), and $(0C, 0D)$ (all opt out).
    The parameter values are $N=100$, $c=1$, $r=1.4$, $\sigma=0.1$, $s=100$, and $\varepsilon=0.01$.}
    \label{fig:2}
\end{figure}

\begin{figure}[!ht]
    \centering
    \includegraphics[width=0.8\textwidth]{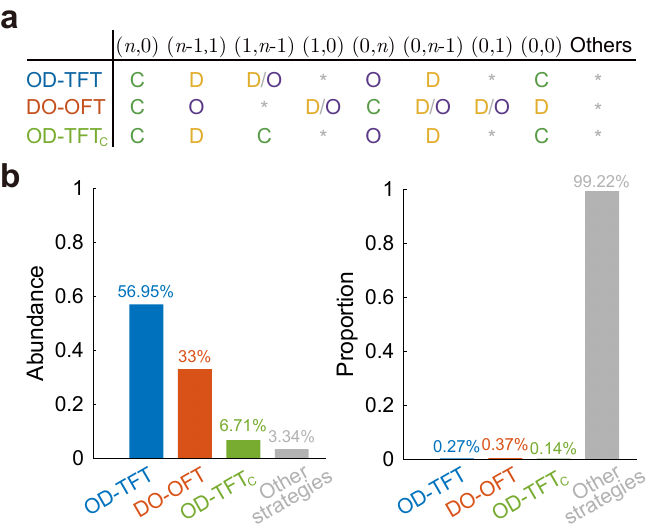}
    \caption{\textbf{Three classes of reactive strategies that readily become equilibria and dominate in the evolutionary dynamics.} \textbf{a}, We characterize all Nash equilibria that support cooperation in pure reactive strategies and identify three novel classes of strategies, $OD$-$TFT$, $DO$-$OFT$, and $OD$-$TFT_C$ (see a full list of equilibria in Supplementary Fig. 2). They are robust to implementation errors and have the lowest threshold for the multiplication factor to surpass to become equilibria. Here, the tuple $(x,y)$ in the first row of each column represents the game state in the previous round, where $x$ and $y$ are numbers of individuals who cooperate and defect, respectively. Symbols $C$, $D$ and $O$ are the actions that $OD$-$TFT$, $DO$-$OFT$ and $OD$-$TFT_C$ prescribe, given the game state $(x,y)$. Symbol $*$ means that $C$, $D$, and $O$ are all allowed. For instance, $OD$-$TFT$ prescribes cooperation after full cooperation (symbol $C$ in the column $(n,0)$) and prescribes opt-out after full defection (symbol $O$ in the column $(0,n)$).
    \textbf{b}, We calculate the abundance of $OD$-$TFT$, $DO$-$OFT$, $OD$-$TFT_C$, and all other strategies in the evolutionary simulations (measured by the average time that each class of strategies occupies the population) and also their proportions within all pure reactive strategies.
    Our results show that these three classes of strategies dominate in the evolutionary dynamics by occupying the population for more than $96\%$ of the time ($56.95\%$ for $OD$-$TFT$, $33\%$ for $DO$-$OFT$ and $6.71\%$ for $OD$-$TFT_C$), with only a proportion of $0.78\%$ ($0.27\%$ for $OD$-$TFT$, $0.37\%$ for $DO$-$OFT$ and $0.14\%$ for $OD$-$TFT_C$) in pure reactive strategies.
    Simulation results are obtained by averaging over $1000$ runs. In each run, $1\times10^5$ mutations are introduced to the population. Other parameter values are the same as those in Fig. \ref{fig:2}.}
\label{fig:strategy}
\end{figure}

\paragraph{Evolutionary advantage under voluntary participation.}
To explore the evolution of cooperation in repeated optional PGGs, we run simulations and analyze a ``melting pot'' of reactive strategies (in total, $3^{(n+1)(n+2)/2}$ strategies).
We find that voluntary participation greatly enhances the cooperation rates in repeated optional PGGs (see the blue line in Fig. \ref{fig:2}a). In contrast, under the same conditions, repeated compulsory PGGs (see the red line in Fig. \ref{fig:2}a) and one-shot optional PGGs (see the orange line in Fig. \ref{fig:2}a) only yield low propensities for cooperation.
This indicates that the combination of voluntary participation and conditional responses is conducive to cooperation.
Here, opt-out acts as a catalyst for the evolution of cooperation: it not only provides a natural way for individuals to escape from the social trap of mutual defection but also serves as a stepping stone to boost cooperation (see Fig. \ref{fig:2}b).

In addition, to understand how individuals react in each game state $(x,y)\!\in\!\mathcal{G}$, we calculate the average tendency for individuals to cooperate, defect, and opt out, i.e., the average strategy (see Fig.~\ref{fig:2}c). Our results show that the average strategy exhibits clear and interesting characteristics: (i) it supports (persistent) cooperation by prescribing cooperation after full cooperation and avoiding persistent defection (i.e., not to defect after full defection) and opt-out (i.e., not to opt out after full opt-out); (ii) it actively leverages opt-out to punish defection, such as when someone in a once fully cooperative group starts to defect (i.e., game state $(2,1)$) or when full defection occurs.

\begin{figure}[!ht]
    \centering
    \includegraphics[width=0.65\textwidth]{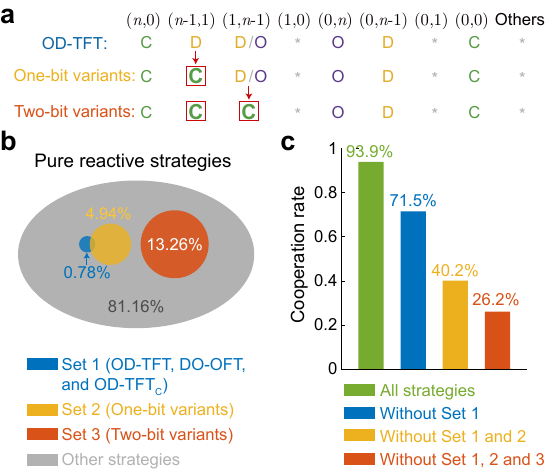}
    \caption{\textbf{$OD$-$TFT$, $DO$-$OFT$, and their behaviorally close variants are vital for the evolution of cooperation.} \textbf{a}, We call a strategy the one-bit variant of $OD$-$TFT$ (or $DO$-$OFT$) if it prescribes different actions from $OD$-$TFT$ (or $DO$-$OFT$) after only one game state $(x,y)$ (except game state $(n,0)$). A strategy is called a two-bit variant of $OD$-$TFT$ (or $DO$-$OFT$) if it prescribes different actions from $OD$-$TFT$ (or $DO$-$OFT$) for two game states (except game state $(n,0)$). Taking $OD$-$TFT$ as an example, one class of its one-bit variants is the strategies that prescribe $D$ after game state $(n-1,1)$ and take the same action as $OD$-$TFT$ after the game state $(1,n-1)$, $(0,n)$, $(0,n-1)$ and $(0,0)$. Based on this class of one-bit variants, a class of $OD$-$TFT$'s two-bit variants is the strategies that additionally prescribe $C$ after game state $(1,n-1)$ and take the same action after the game state $(0,n)$, $(0,n-1)$ and $(0,0)$. Note that $OD$-$TFT_C$ is one of the one-bit variants of $OD$-$TFT$.
    \textbf{b}, $OD$-$TFT$, $DO$-$OFT$ and their behaviorally close variants make up less than $20\%$ ($0.78\%$ for $OD$-$TFT$, $DO$-$OFT$ and $OD$-$TFT_C$, $4.94\%$ for their one-bit variants and $13.26\%$ for their two-bit variants) of all pure reactive strategies. For better illustrations, we denote the set of strategies included in $OD$-$TFT$, $DO$-$OFT$ and $OD$-$TFT_C$ as Set 1, the set included in the one-bit variants of $OD$-$TFT$ and $DO$-$OFT$ as Set 2, and that in the two-bit variants as Set 3. Since $OD$-$TFT_C$ is a one-bit variant of $OD$-$TFT$, Set 1 and Set 2 intersect at $OD$-$TFT_C$.
    \textbf{c}, To find which strategies are key to the evolution of cooperation in repeated optional PGGs, we conduct a series of ``knock-out experiments'', in which we intentionally and progressively delete strategies and see how this alters the cooperation rate. When knocking out Set 1, we find that the cooperation rate falls from $93.9\%$ to $71.5\%$, indicating that the repeated optional games can no longer yield a high level of cooperation without Set 1. When we further delete Set 1 and Set 2, the cooperation rate falls sharply to only $40.2\%$.
    Simulation settings and parameter values are the same as those in Fig. \ref{fig:2}.}
\label{fig:KnockOutExperiment}
\end{figure}
\bigskip

\paragraph{Equilibrium analysis for repeated optional PGGs.}
Based on the characteristics reflected by the average strategy, we turn to strategies that support cooperation and try to identify key strategies therein that promote the evolution of cooperation in repeated optional PGGs.
Due to the presence of implementation errors, such strategies are expected to be error-robust, which means mutual cooperation is not undermined by occasional errors if these strategies are used by all individuals in the group.
Besides, some kind of stability is also needed to ensure that these strategies are not easily invaded.
Here, we mainly consider the stability of strategies imposed by being a Nash equilibrium since evolutionary stability is generally not attainable in repeated games~\cite{Boyd1987,Bendor1995,Garcia2018} and seems less important with a large strategy space~\cite{Li2022}.
Indeed, previous studies about repeated compulsory PGGs find that a high level of cooperation is often reached if the all-or-none strategy ($AoN$) is an equilibrium; otherwise, defection is favored~\cite{Pinheiro2014, Hilbe2017}.
Specifically, $AoN$ becomes an equilibrium if $r \ge 2n/(n+1)$ \cite{Hilbe2017,hilbe2014cooperation}.

In repeated optional PGGs, to provide a complete characterization of Nash equilibria that support cooperation, we first note that if all individuals use the same pure reactive strategy, they always take the same action after any game state. Employing this, we are able to characterize all possible equilibria that support cooperation in reactive strategies (see a full list of equilibria in Supplementary Fig. 2).
By analytically deriving their associated conditions to become equilibria, we find that three novel classes of strategies, $OD$-$TFT$, $DO$-$OFT$, and $OD$-$TFT_C$ (see descriptions of these strategies in Fig.~\ref{fig:strategy}a), have the lowest threshold for the multiplication factor $r$ to surpass within all strategies that can robustly support cooperation. More importantly, this threshold to become an equilibrium is lower than that for $AoN$. In detail, there exists a region of parameters where these strategies are equilibria while $AoN$ is not. Such a region for $OD$-$TFT$, $DO$-$OFT$, and $OD$-$TFT_C$ to become equilibria is
\begin{equation}\label{conditionOptionalGames: nash}
  \frac{3n}{2n+1}\le r < \frac{2n}{n+1} \text{ and } \sigma \le \frac{2n+1}{n}rc - 3c.
\end{equation}
Note that condition (\ref{conditionOptionalGames: nash}) is not only related to $r$ but also the payoff for opt-out, i.e., $\sigma$, meaning that the superiority of $OD$-$TFT$, $DO$-$OFT$, and $OD$-$TFT_C$ over $AoN$ in the condition to become an equilibrium is realized by the additional option to opt out.

To intuitively understand this superiority, we offer the following explanations. The common characteristics of $OD$-$TFT$, $DO$-$OFT$, and $OD$-$TFT_C$ are that: (i) they keep cooperating when all individuals cooperate in the previous round, (ii) they correct errors and recover cooperation within at most three rounds, and (iii) they respond to defection by a series of ordered actions of opt-out and defection (or defection and opt-out). The third characteristic indicates, compared with $AoN$, $OD$-$TFT$, $DO$-$OFT$, and $OD$-$TFT_C$ can pose more threats to individuals who intend to defect since individuals adopting these strategies can respond to defection by not only defecting but also opting out. Such combined reactions further reduce the net gains (or even make them negative) for individuals who deviate by defecting. This eventually leads to that $OD$-$TFT$, $DO$-$OFT$, and $OD$-$TFT_C$ have a lower threshold for $r$ to surpass than $AoN$ does.

The above analysis demonstrates the superiority of $OD$-$TFT$, $DO$-$OFT$, and $OD$-$TFT_C$ in the static sense. Here, we further test their performance in the evolutionary dynamics. To this end, we calculate the average abundance of $OD$-$TFT$, $DO$-$OFT$, $OD$-$TFT_C$, and all other strategies. We find that these three classes of strategies dominate in the evolutionary dynamics, occupying the population for more than $96\%$ of the time ($56.95\%$ for $OD$-$TFT$, $33\%$ for $DO$-$OFT$, and $6.71\%$ for $OD$-$TFT_C$) with only a proportion less than $1\%$ ($0.27\%$ for $OD$-$TFT$, $0.37\%$ for $DO$-$OFT$, and $0.14\%$ for $OD$-$TFT_C$) (see Fig.~\ref{fig:strategy}b). This indicates that the prediction of the static equilibrium analysis is in well agreement with that of the evolutionary dynamics.
\bigskip

\paragraph{Key strategies for the evolutionary advantage.}
Generally, to tell which strategies are key to the evolution of cooperation, a good approach is to conduct the so-called ``knock-out experiments''~\cite{imhof2010stochastic}, in which we intentionally and progressively delete strategies and see how this alters the cooperation rate. If the cooperation rate drops substantially after deletion, it means that the strategies deleted are crucial to the evolution of cooperation.

We start our ``knock-out experiments'' by deleting $OD$-$TFT$, $DO$-$OFT$, and $OD$-$TFT_C$ that dominate in the evolutionary dynamics when they are present. Our results show that the roles that $OD$-$TFT$, $DO$-$OFT$, and $OD$-$TFT_C$ play in the evolution of cooperation is to maintain a high level of cooperation (greater than $90\%$). If they are deleted, the level of cooperation falls to $71.5\%$ (see Fig.~\ref{fig:KnockOutExperiment}c). Despite quite a fall, the cooperation rate is still much higher than that in the repeated compulsory PGGs. To further identify strategies that account for the significant enhancement of cooperation in repeated optional PGGs, we turn to strategies that are not equilibria but have very similar behavioral patterns to $OD$-$TFT$, $DO$-$OFT$, and $OD$-$TFT_C$. This leads us to strategies that are one-bit and two-bit different from $OD$-$TFT$ and $DO$-$OFT$ (see Fig.~\ref{fig:KnockOutExperiment}a for illustrations; note that $OD$-$TFT_C$ is actually one of the one-bit variants of $OD$-$TFT$). In the second and third experiments, we delete not only $OD$-$TFT$ and $DO$-$OFT$, but also these close variants (including $OD$-$TFT_C$). Our results indicate that deleting $OD$-$TFT$, $DO$-$OFT$, and their one-bit variants decreases the cooperation rate substantially (by about $54\%$, from about $94\%$ to $40\%$), and the cooperation rate further drops by about $14\%$ if the two-bit variants are deleted. This reveals that $OD$-$TFT$ and $DO$-$OFT$ and their behaviorally close variants are key to the evolution of cooperation in repeated optional PGGs.

\begin{figure}[!ht]
    \centering
    \includegraphics[width=\textwidth]{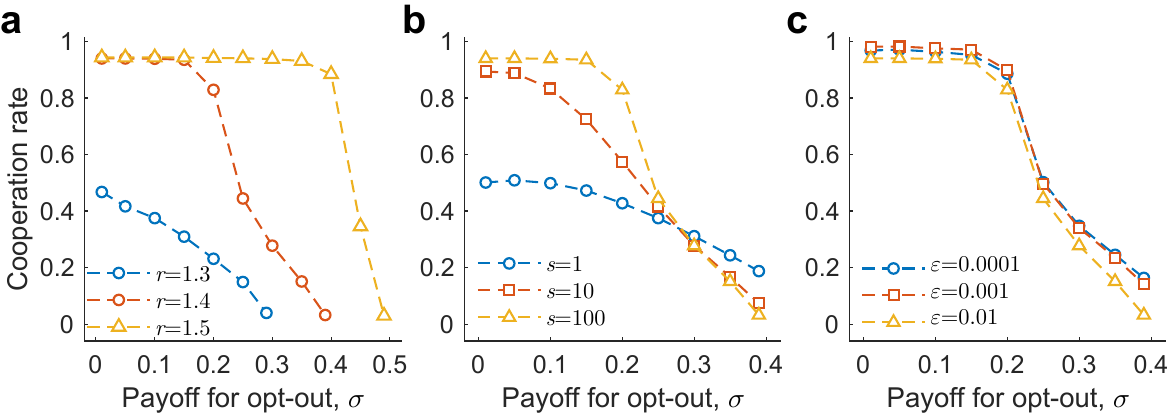}
    \caption{\textbf{A small incentive for opt-out is enough to achieve a high level of cooperation.} We plot the cooperation rate as a function of the payoff for opt-out, $\sigma$, under various multiplication factors $r$ (\textbf{a}), selection intensities $s$ (\textbf{b}), and error rates $\varepsilon$ (\textbf{c}). Our results show that to achieve a high level of cooperation, a small incentive for opt-out is enough when the multiplication factor surpasses a threshold ($r\ge 1.4$), the selection intensity is sufficiently strong ($s\ge 10$), and errors are infrequent ($\varepsilon \le 0.01$). In addition, if the payoff $\sigma$ for opt-out is too large, it inhibits the evolution of cooperation by encouraging individuals to opt out. Other parameter values except for $\sigma$ are the same as those in Fig. \ref{fig:2}.}
    \label{fig:sigma}
\end{figure}
\bigskip

\paragraph{The effect of the payoff for opt-out on cooperation.}
In our previous investigations, we find that allowing individuals to opt out significantly promotes the evolution of cooperation and we also identifies key strategies that account for such a promotive effect. A natural follow-up question would be to what extent, this promotive effect is affected by the incentive to opt out, namely, the payoff for opt-out, $\sigma$.
To answer this question, we consider the effect of $\sigma$ on the evolution of cooperation under various situations, including different multiplication factors, selection intensities, and error rates (see Fig. \ref{fig:sigma}).

Our results show that a small incentive for opt-out is enough to achieve a high level of cooperation if the multiplication factor is above a threshold ($r\ge 1.4$ in Fig. \ref{fig:sigma}a), the selection intensity is sufficiently strong ($s\ge 10$ in Fig. \ref{fig:sigma}b), and errors are infrequent ($\varepsilon \le 0.01$ in Fig. \ref{fig:sigma}c).
Moreover, we find that if the incentive for opt-out is too large, it is actually detrimental to the evolution of cooperation. The reason is that individuals now would prefer opting out and gaining a decent and safe income instead of taking the risky action of cooperation. Our results reveal that the highly cooperative population achieved in repeated optional PGGs depends largely on the permission to opt out itself while the selfish drive to gain a high profit by opting out is less relevant and even inhibits cooperation. This echoes our findings through equilibrium analysis: the smaller the payoff for opt-out is, the more severe the punishment is for individuals who deviate from mutual cooperation if others are using the equilibrium strategy such as $OD$-$TFT$, $DO$-$OFT$, or $OD$-$TFT_C$.

\begin{figure}[!ht]
    \centering
    \includegraphics[width=0.7\textwidth]{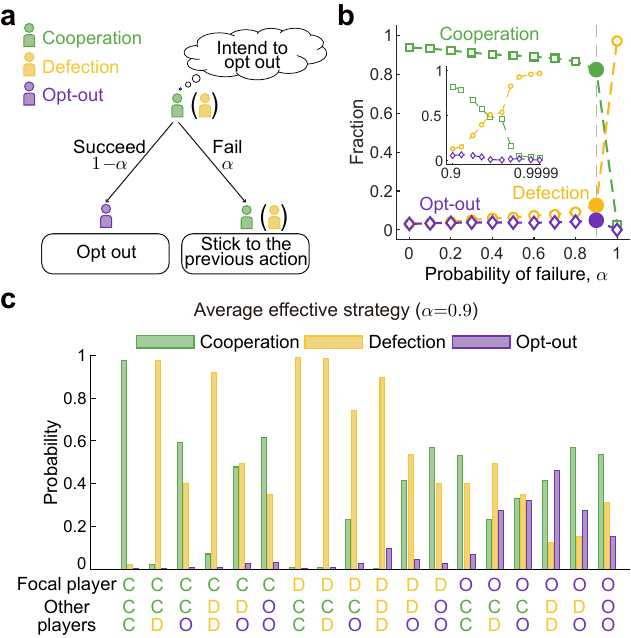}
    \caption{\textbf{Our findings are robust to the failure probability of opting out.} \textbf{a}, We consider repeated optional PGGs with possible failures of opting out. In such a game, an individual who intends to opt out in the next round will successfully do so with probability $1-\alpha$. With probability $\alpha$, it fails to opt out and sticks to its previous action. \textbf{b}, We simulate the evolutionary dynamics of such repeated optional PGGs for different failure probabilities $\alpha$. Our results show that a high level of cooperation persists for a wide range of $\alpha$ values (only substantially decreases when $\alpha>0.9$), indicating that the effect of voluntary participation on cooperation is robust to the failure probability of opting out. \textbf{c}, Further analysis on the effective strategy that dominates in the evolutionary dynamics under $\alpha=0.9$ confirms this result. Each bar represents the probability to cooperate, defect, and opt out, given that the previous action profile is $(A, BC)$ ($A$ is the action of the focal individual, $BC$ are the actions of the others). The strategy is the average effective strategy over 1000 simulations. In each simulation, $1\times10^5$ mutations are introduced to the population. We find that the dominant strategy is in some way similar to $OD$-$TFT$. Yet, compared to repeated optional PGGs which allow freely dropping out, the dominant strategy evolves an additional characteristic to restore cooperation quickly. Other parameters are the same as those in Fig. \ref{fig:2}.}
    \label{fig:risk}
\end{figure}
\bigskip

\paragraph{Repeated optional PGGs with failures of opting out.}
Beyond the basic model, we also investigate repeated optional PGGs with possible failures of opting out, where individuals who intend to drop out from the interaction may fail to do so and instead stick to its previous action (see Fig.~\ref{fig:risk}a for illustrations).
We denote the probability of such failure as $\alpha \in [0,1]$. When $\alpha = 0$, it means that individuals can freely opt out of the interaction, which reduces to the basic model. When $\alpha = 1$, individuals will always fail to opt out and the game becomes compulsory. When $\alpha \in (0,1)$, there is always a positive probability that individuals fail to opt out. This describes a situation where the action of opt-out takes a delayed effect or opting out of the interaction becomes restricted. Intuitively, it is expected that the cooperation rate will decrease as $\alpha$ increases in a steady way. Counterintuitively, our simulation results reveal that the cooperation rate drops gradually (and remains at a high level) when $\alpha<0.9$ but substantially as $\alpha>0.9$ (see Fig.~\ref{fig:risk}b). This means that the collapse of cooperation only becomes prominent when it is very much likely that individuals will fail to opt out.

To better illustrate the underlying mechanisms, we plot the average effective strategy under a high probability of opt-out failure ($\alpha=0.9$) in Fig.~\ref{fig:risk}c. We find that the average effective strategy (by considering the failure of opt-out) in some way behaves similarly to $OD$-$TFT$: it prescribes cooperation after full cooperation, it defects when one defects and everyone else cooperates, it tries to opt out if all defect and to restore cooperation if everyone drops out. In addition, due to the prevailing failures of opting out, the strategy possesses its own characteristics to restore cooperation quickly (see Supplementary Section 5.1 for details). In the case where all individuals defect in the previous round (see ($D$, $DD$) in Fig.~\ref{fig:risk}c), individuals attempt to opt out of the game. When someone opts out successfully after full defection, all individuals will take action $C$ at the same time or keep attempting to opt out (see ($O$, $DD$) and ($D$, $DO$) in Fig.~\ref{fig:risk}c). (Note that individuals take the same action after game states ($O$, $DD$) and ($D$, $DO$) when adopting the same pure reactive strategy.) When two of them opt out successfully, individuals directly return to the state of mutual cooperation by cooperating or indirectly restore cooperation by sticking to opt out and cooperating after game state ($O$, $OO$) (see ($O$, $DO$), ($D$, $OO$) and ($O$, $OO$) in Fig.~\ref{fig:risk}c). Such a quick path to restore cooperation opens up possibilities for the evolution of cooperation.

Besides, we thoroughly test the robustness of our findings in the space of stochastic and memory-one strategies, to initial strategy configurations of the population, and to various other model parameters, including multiplication factors, selection intensities, error rates, and discount factors (see Supplementary Figs. 1 and 3). The simulation results show that voluntary participation robustly promotes the evolution of cooperation under a wide range of model settings. In addition, we provide detailed investigations about the feasibility of cooperation in one-shot optional PGGs in finite populations (see Supplementary Section 5.3 for details) and find that the maximum possible level of cooperation is less than $40\%$ with the same set of parameters (except $r$ and $\sigma$) as that in Fig. 2, implying that one-shot optional PGGs can never yield a cooperation rate that is comparable to repeated optional PGGs (see Supplementary Fig. 4). At last, similar results are also found in repeated optional prisoner's dilemma games (see Supplementary Section 4 and Supplementary Fig. 5 for details), demonstrating that our finding also apply to prisoner's dilemma games.

\section*{Discussion}
In this work, we introduce a general framework of repeated optional PGGs where individuals can opt out of each interaction and resume participation later. This allows us to investigate how voluntary participation affects the evolution of cooperation in repeated interactions.
Under our framework, we show that repeated optional PGGs lead to almost full cooperation even in situations where repeated compulsory PGGs and one-shot optional PGGs yield low levels of cooperation.
This hinges on the synergistic interplay between optional participation and conditional responses. Once individuals who engage in repeated optional PGGs are forbidden to opt out, the cooperation rate plunges (see Supplementary Fig. 6), with later recovery if individuals are allowed to opt out again.
Moreover, we find that three novel classes of reactive strategies, $OD$-$TFT$, $DO$-$OFT$, and $OD$-$TFT_C$ become Nash equilibria with a lower threshold for the multiplication factor than $AoN$.
These strategies prescribe cooperation if all individuals do so in the previous round, they correct errors and restore cooperation within at most three rounds, and they respond to defection by a series of ordered actions of opt-out and defection (or defection and opt-out).
Together with their dominance in the evolutionary dynamics, $OD$-$TFT$, $DO$-$OFT$, and $OD$-$TFT_C$ are the main contributors to the high level of cooperation in repeated optional PGGs.
Besides, by conducting a series of ``knock-out experiments'', we reveal that $OD$-$TFT$, $DO$-$OFT$, and their behaviorally close variants ($OD$-$TFT_C$ included) are key to the evolution of cooperation, accounting for the evolutionary advantage under voluntary participation.
Our results thus indicate that voluntary participation is a simple and effective mechanism for the evolution of cooperation in repeated interactions.

Previous studies on one-shot, non-repeated, optional PGGs show that when the multiplication factor is less than the minimum group size (i.e., $r<2$), evolutionary dynamics eventually lead to full opt-out in infinite populations \cite{Hauert2002,Hauert2002a}. Only when the multiplication factor is no less than two, cooperation becomes feasible. The reason for this is that for a fixed multiplication factor $r$ $(\ge 2)$, the actual size of the PGG is reduced when individuals opt out; this further leads to a game that may no longer be a social dilemma, and cooperation becomes appealing~\cite{Hauert2002a}. We stress that the promotion of cooperation in repeated optional PGGs does not rely on such a mechanism: throughout the main text, $r<2$ is enforced to guarantee that every single PGG played is a social dilemma. This highlights that repeated interaction is important for facilitating cooperation in optional interactions.

In another strand of existing literature on voluntary participation, opting out of the interaction means terminating the game with the current partners and then starting another game with new ones~\cite{Izquierdo2010, Izquierdo2014, Kurokawa2019, Kurokawa2021, Kurokawa2022}. It is found that this promotes the evolution of cooperation. However, we note that the underlying mechanism for the evolution of cooperation in these studies is conditional dissociation, which is also different from ours. Under conditional dissociation, cooperators will terminate interactions with defectors and interact more often with cooperators. This process essentially leads to the assortment of cooperators, which thus facilitates the evolution of cooperation. However, in our work, conditional cooperation and the threat of nonparticipation, instead of assortment, are key to the evolution of cooperation.

In addition, the three novel classes of strategies dominating in the evolutionary dynamics, $OD$-$TFT$, $DO$-$OFT$, and $OD$-$TFT_C$, punish intentional defection like memory-two strategies do in repeated compulsory PGGs. Individuals who adopt memory-two strategies use action profiles of the previous two rounds to determine the next move. It is thus possible for them to retaliate against defection for two successive rounds, which reduces the incentives to deviate from mutual cooperation and is an essential characteristic for successful strategies (e.g., memory-two all-or-none strategy, $AoN_2$)~\cite{Hilbe2017}. Despite having a shorter memory, $OD$-$TFT$, $DO$-$OFT$, and $OD$-$TFT_C$ are also able to achieve this type of retaliation through a series of ordered actions of opt-out and defection (or defection and opt-out).
This similarity suggests that the advantage of a long memory may be in some sense realized by extending the set of available actions. This inspires us to explore whether the level of cooperation can be further improved by providing more available actions to opt out (see Supplementary Section 5.2 for details). Our results show that repeated PGGs with two additional actions of opting out indeed further increase the level of cooperation. Our findings imply that the ubiquity of cooperation in animals with low cognitive capabilities may result from a large set of available actions (e.g., hard-wired reactions).

In most modern human activities, voluntary participation is a basic right. Our work suggests that the significance of voluntary participation manifests itself not only as a mark of freedom but also as an effective means to promote human cooperation. It also indicates that previous models of compulsory games may underestimate the capability of humans to foster reciprocal cooperation, and opt-out may serve as new leverage to enforce cooperation. Due to the importance of voluntary participation, it would be interesting to explore how voluntary participation itself evolves and how the coevolution of opt-out and cooperation affects the evolutionary outcomes, which may help us better understand the evolution of human cooperation.



\pagebreak

\begin{center}
    \Large Supplementary Information for \\ \textbf{Evolutionary dynamics in repeated optional games} \\[.2cm]
    Fang Chen, Lei Zhou, Long Wang\\[.1cm]
\end{center}

\setcounter{equation}{0}
\setcounter{figure}{0}
\setcounter{table}{0}
\setcounter{page}{1}
\setcounter{section}{0}
\renewcommand{\theequation}{S\arabic{equation}}

\section{Overview}
In the following, we provide detailed derivations for the results in the main text. We first describe our model and methods in Section \ref{sec:modelandmethod}. Then, in Section \ref{sec:equilibriumanalysis}, we give an equilibrium analysis for repeated optional public goods games. We identify all pure reactive strategies that give rise to persistent cooperation, defection, and opt-out, respectively. Finally, in Section \ref{sec:further}, we investigate the evolution of cooperation in two extended models and analyze the feasibility of cooperation in one-shot optional PGGs. All proofs of our theorems are presented in the Appendix.

\section{Model and method}\label{sec:modelandmethod}
\subsection{Game setup}\label{sec:gamesetup}
We consider a repeated optional \textit{public goods game} (PGG) in a group with $n$ individuals. In each round, individuals can decide to participate in the game and cooperate ($C$), to participate in the game and defect ($D$), or to opt out of the game ($O$). Among all participants of the game, individuals who cooperate contribute an endowment, $c$, to the public good, and those who defect contribute nothing. The total contributions are then multiplied by a multiplication factor $1<r<n$, and uniformly divided among all participants. In the meanwhile, individuals who opt out of the game obtain a fixed payoff, $\sigma$, regardless of the number of individuals who cooperate, defect and opt out. For instance, if $x$ individuals cooperate, $y$ individuals defect, and there are at least two individuals participating in the game (i.e., $1<x+y\le n$), the individuals who cooperate get a payoff $rxc/(x+y)-c$, those who defect $rxc/(x+y)$, and those who opt out $\sigma$. Note that if $x+y=1$, namely, only one individual participates in the game, the game will be canceled and this individual gets a payoff of $\sigma$. In repeated optional PGGs, we assume $0<\sigma<(r-1)c$ such that full cooperation is better off than opting out, and opting out is better off than full defection~\cite{Hauert2002,Hauert2002a}.

In our work, we consider repeated optional PGGs played for infinitely many rounds and the future payoffs are discounted by a factor $\delta$ with $0<\delta<1$. This model setting can also be interpreted as an indefinitely repeated game in which the next round occurs with a probability $\delta$. In the main text and Section \ref{sec:equilibriumanalysis}, we focus on the limiting case of no discounting, $\delta\rightarrow 1$.

\subsection{Reactive strategies}
Strategies for repeated optional PGGs can be arbitrarily complex: in general, they can take all individuals' historical actions as input, and return probabilities to cooperate, defect, and opt out as output. To make the evolutionary and theoretical analysis feasible, we restrict ourselves to reactive strategies, with which an individual determines its action based on the number of individuals who cooperate and defect in the previous round~\cite{Pinheiro2014}. Here, we use $x$ ($y$) to denote the number of individuals who cooperate (defect) in the previous round. The game state in each round can be represented as a tuple $(x,y)$ where $0 \leq x+y \leq n$ is the number of participants and $n-x-y$ is the number of individuals who opt out. There are in total $\sum_{x=0}^{n}\sum_{y=0}^{n-x}1=(n+1)(n+2)/2$ game states. Formally, we represent a reactive strategy as a $2[(n+1)(n+2)+1]$-dimensional vector
\begin{equation}
    \mathbf{p}=(p_0^C,p_0^D;p_{x,y}^C,p_{x,y}^D),
    \label{equ:strategy}
\end{equation}
where $p_0^C$ ($p_0^D$) corresponds to the probability to cooperate (defect) in the first round, and $p_{x,y}^C$ ($p_{x,y}^D$) is the probability that an individual cooperates (defects), given that the game state in the previous round is $(x,y)$. Correspondingly, the probability $p_{x,y}^O$ that the individual opts out of the game is $1-p_{x,y}^C-p_{x,y}^D$. We say that a strategy $\mathbf{p}$ is \textit{pure} (\textit{deterministic}) if all entries of $\mathbf{p}$ are either zero or one; otherwise, it is \textit{stochastic}.

By setting restrictions on the strategies that individuals can choose, our model can recover the game dynamics in traditional repeated compulsory PGGs and one-shot optional PGGs as special cases. For repeated compulsory PGGs, we achieve this by forbidding the additional option of opting out, namely, by setting $p_{x,y}^C+p_{x,y}^D\equiv1$ (or $p_{x,y}^O \equiv 0$). Of course, the classical all-or-none strategy ($AoN$)~\cite{Pinheiro2014,Hilbe2017}, with which an individual cooperates if all individuals take the same action and defects otherwise, can be readily represented by setting $p_{n,0}^C=p_{0,n}^C=p_{0,0}^C=1$ and $p_{x,y}^D=1$ for all $(x,y)\notin\{(n,0),(0,n),(0,0)\}$. To recover the game dynamics in one-shot optional PGGs, we restrict the set of available strategies, leading to a set with only three strategies that contains always cooperating ($ALLC$, $p_{x,y}^C=1$ for any $x$ and $y$), always defecting ($ALLD$, $p_{x,y}^D=1$ for any $x$ and $y$) and always opting out ($ALLO$, $p_{x,y}^O=1$ for any $x$ and $y$).

In addition, we also assume that the individuals have a ``trembling hand'' or are subject to implementation errors such that they cannot execute their actions perfectly. For instance, when an individual decides to cooperate, there is a probability $\varepsilon>0$ that it implements a wrong action. Unbiasedly, the individual instead performs the other two actions with equal probability, i.e., to defect or to opt out with probability $\varepsilon/2$. Formally, for an individual with strategy $\mathbf{p}$, such a ``trembling hand'' results in an effective strategy $(1-\varepsilon)\mathbf{p}+\frac{\varepsilon}{2}(1-\mathbf{p})$. Note that if individuals use pure strategies, implementation errors make their effective strategies stochastic. In this case, long-term payoffs in an infinitely repeated game become independent of the move in the very first round. Therefore, we only need to consider simpler strategies with $\mathbf{p}=(p_{x,y}^C,p_{x,y}^D)$.

\subsection{Long-term payoffs}\label{sec:tranditionalmethod}
We denote the action profile in one round of the repeated games as $(A_1,...,A_n)$, where $A_i$ is the action of individual $i$ in that round. There are in total $3^{n}$ action profiles since each entry of the profile can be $C$, $D$, or $O$. If all individuals adopt reactive strategies, we can use a Markov chain to model the dynamics of play in a repeated optional PGG. The states of the Markov chain are all possible action profiles. Suppose that in a repeated optional PGG, individual $i$'s strategy is $\mathbf{p}^{(i)}$. Then the transition probability for the action profile to change from $(A_1,...,A_n)$ to $(A'_1,...,A'_n)$ in the next round is
\begin{equation}
    m_{(A_1,...,A_n)\rightarrow(A'_1,...,A'_n)}=\prod_{i=1}^n(p_{x,y}^{A'_i})^{(i)},
    \label{equ:transition_matrix}
\end{equation}
where $x$ and $y$ are, respectively, the number of $C$s and $D$s in $(A_1,...,A_n)$, and $(p_{x,y}^{A'_i})^{(i)}$ is the associated probability that individual $i$ implements action $A'_i\in\{C,D,O\}$. Since individuals make their decisions independently, the transition probability $m_{(A_1,...,A_n)\rightarrow(A'_1,...,A'_n)}$ is a product of the probabilities that each individual implements $A'_i$. Similarly, the probability that the action profile $(A_1,...,A_n)$ occurs in the very first round is
\begin{equation}
    v_{A_1,...,A_n}(0)=\prod_{i=1}^n (p_0^{A_i})^{(i)}.
    \label{equ:initial_distribution}
\end{equation}

To calculate the long-term payoffs, let us collect all probabilities in Eq.~\eqref{equ:transition_matrix} and build a $3^n \times 3^n$ transition matrix $\mathbf{M}$. We also use a row-vector $\mathbf{v}(t)$ to denote the probabilities that each action profile occurs in the first round (see Eq.~\eqref{equ:initial_distribution}) and denote $\mathbf{v}(t)=\mathbf{v}(0)\mathbf{M}^t$ as the distribution over all states in round $t$.

When future payoffs are not discounted, namely, $\delta = 1$, the long-term payoff is calculated as the average payoff per round. To this end, we need to compute the average distribution
\begin{equation}
    \bar{\mathbf{v}}=\lim\limits_{T\rightarrow \infty }\frac{1}{T}\sum_{t=0}^{T}\mathbf{v}(t) = \lim\limits_{T\rightarrow \infty }\frac{1}{T}\sum_{t=0}^{T}\mathbf{v}(t)\mathbf{M}^t.
\end{equation}
Under the assumption of ``trembling hands'', the average distribution $\bar{\mathbf{v}}$ defined in Eq.~(\ref{equ:average_distribution_noDiscounting}) always exists and is unique. And it equals the stationary distribution $\mathbf{v}$ of the Markov chain
\begin{equation}
    \mathbf{v}=\mathbf{vM}.
    \label{equ:average_distribution_noDiscounting}
\end{equation}

When future payoffs are discounted with $0<\delta<1$, the long-term payoff is calculated as the discounted payoff. The distribution becomes
\begin{equation}
    \tilde{\mathbf{v}}=\lim\limits_{T\rightarrow \infty }(1-\delta)\sum_{t=0}^{T}\delta^t \mathbf{v}(t)=(1-\delta)\mathbf{v}(0) (\mathbf{I}-\delta \mathbf{M})^{-1},
    \label{equ:average_distribution}
\end{equation}
where $\mathbf{I}$ is the identity matrix with a suitable dimension. The entries $\bar{v}_{A_1,...,A_n}$ ($\tilde{v}_{A_1,...,A_n}$) of the above vector can be regarded as the probability that one finds itself in state $(A_1,...,A_n)$ over the course of the play.

Let $\mathcal{V}_{m}$ denote the state where $m$ individuals participate in the game, i.e.
\begin{equation*}
    \mathcal{V}_m=\{(A_1,...,A_n)~|~\text{the number of $O$s is } n-m \}.
\end{equation*}
Furthermore, we define the subset of $\mathcal{V}_m$ in which individual $i$ takes action $A$ and $n_c$ individuals cooperate as
\begin{equation*}
    \mathcal{V}_{A,m,n_c}^i=\{(A_1,...,A_n)|A_i=A \text{ and, there are } n_c \text{ $C$s and } n-m \text{ $O$s}\}.
\end{equation*}
Then, the long-term payoff of individual $i$ can be calculated as
\begin{equation}
    \begin{aligned}
    \pi^{(i)}=&\sum_{m=2}^n\sum_{n_c=0}^m\left(\sum_{(A_1,...,A_n)\in\mathcal{V}_{C,m,n_c}^i}(\frac{n_crc}{m}-c)\hat{v}_{A_1,...,A_n}+\sum_{(A_1,...,A_n)\in\mathcal{V}_{D,m,n_c}^i}\frac{n_crc}{m}\hat{v}_{A_1,...,A_n}\right)\\
    &+\left(\sum_{(A_1,...,A_n)\in\mathcal{V}_0}\hat{v}_{A_1,...,A_n}+\sum_{(A_1,...,A_n)\in\mathcal{V}_1 1}\hat{v}_{A_1,...,A_n}\right)\sigma,
    \end{aligned}
    \label{equ:long_term_payoff}
\end{equation}
where $\hat{v}_{A_1,...,A_n}=\bar{v}_{A_1,...,A_n}$ if $\delta = 1$ and $\hat{v}_{A_1,...,A_n}=\tilde{v}_{A_1,...,A_n}$ if $0 < \delta < 1$. In the limit of no discounting $\delta\rightarrow 1$, Eq.~(\ref{equ:long_term_payoff}) with $0<\delta<1$ yields the same payoff as that with $\delta = 1$ for any individual $i$.

\subsection{State-clustering method}\label{sec:stateclustering}
To further reduce the time for numerically calculating the long-term payoffs mentioned above, we introduce the state-clustering method~\cite{Chen2022my}. Suppose there are two strategies, $\mathbf{p}$ and $\mathbf{q}$, in the group. Employing symmetry of the game, it is enough to calculate the expected payoffs for  $\mathbf{p}$ ($\mathbf{q}$) players if the number of $\mathbf{p}$ ($\mathbf{q}$) players who cooperate, defect and opt out are known. Based on this, the state-clustering method aggregates action profiles by the number of individuals who cooperate, defect and opt out among $\mathbf{p}$ players and those among $\mathbf{q}$ players. Let $x_1$, $y_1$ and $z_1$ denote the numbers of $\mathbf{p}$ individuals that cooperate, defect and opt out in $(A_1,...,A_n)$, respectively. Correspondingly, let $x_2$, $y_2$ and $z_2$ denote the numbers of $\mathbf{q}$ individuals who cooperate, defect and opt out. The state aggregated under the state-clustering method can be represented as $C_{x_1}^{\mathbf{p}}D_{y_1}^{\mathbf{p}}O_{z_1}^{\mathbf{p}}C_{x_2}^{\mathbf{q}}D_{y_2}^{\mathbf{q}}O_{z_2}^{\mathbf{q}}$. Suppose there are $k$ individuals adopting $\mathbf{p}$ and $n-k$ individuals adopting $\mathbf{q}$. There are in total $(k+1)(k+2)(n-k+1)(n-k+2)/4$ states. If the state in round $t$ is $C_{x_1}^{\mathbf{p}}D_{y_1}^{\mathbf{p}}O_{z_1}^{\mathbf{p}}C_{x_2}^{\mathbf{q}}D_{y_2}^{\mathbf{q}}O_{z_2}^{\mathbf{q}}$, the individual adopting $\mathbf{p}$ ($\mathbf{q}$) takes action $A$ in round $t+1$ with probability $p_{x_1+x_2,y_1+y_2}^A$ ($q_{x_1+x_2,y_1+y_2}^A$). Thus, the probability to move from $C_{x_1}^{\mathbf{p}}D_{y_1}^{\mathbf{p}}O_{z_1}^{\mathbf{p}}$ to $C_{x'_1}^{\mathbf{p}}D_{y'_1}^{\mathbf{p}}O_{z'_1}^{\mathbf{p}}$ is
\begin{equation}
    r_{C_{x_1}^{\mathbf{p}}D_{y_1}^{\mathbf{p}}O_{z_1}^{\mathbf{p}}\rightarrow C_{x'_1}^{\mathbf{p}}D_{y'_1}^{\mathbf{p}}O_{z'_1}^{\mathbf{p}}}=
    \tbinom{k}{x'_1}\left(p_{x_1+x_2,y_1+y_2}^C\right)^{x_1'}\tbinom{k-x_1'}{y_1'}
    \left(p_{x_1+x_2,y_1+y_2}^D\right)^{y_1'}
    \left(p_{x_1+x_2,y_1+y_2}^O\right)^{k-x_1'-y_1'}
    \label{equ:reduced_transitionp}
\end{equation}
and
the probability to move from $C_{x_2}^{\mathbf{q}}D_{y_2}^{\mathbf{q}}O_{z_2}^{\mathbf{q}}$ to $C_{x'_2}^{\mathbf{q}}D_{y'_2}^{\mathbf{q}}O_{z'_2}^{\mathbf{q}}$ is
\begin{equation}
    r_{C_{x_2}^{\mathbf{q}}D_{y_2}^{\mathbf{q}}O_{z_2}^{\mathbf{q}}\rightarrow C_{x'_2}^{\mathbf{q}}D_{y'_2}^{\mathbf{q}}O_{z'_2}^{\mathbf{q}}}=
    \tbinom{n-k}{x'_2}\left(q_{x_1+x_2,y_1+y_2}^C\right)^{x_2'}\tbinom{n-k-x_2'}{y_2'}
    \left(q_{x_1+x_2,y_1+y_2}^D\right)^{y_2'}
    \left(q_{x_1+x_2,y_1+y_2}^O\right)^{k-x_2'-y_2'}.
    \label{equ:reduced_transitionq}
\end{equation}
Combining Eq.~\eqref{equ:reduced_transitionp} and Eq.~\eqref{equ:reduced_transitionq}, the transition probability that the state move from $C_{x_1}^{\mathbf{p}}D_{y_1}^{\mathbf{p}}O_{z_1}^{\mathbf{p}}C_{x_2}^{\mathbf{q}}D_{y_2}^{\mathbf{q}}O_{z_2}^{\mathbf{q}}$ to $C_{x'_1}^{\mathbf{p}}D_{y'_1}^{\mathbf{p}}O_{z'_1}^{\mathbf{p}}C_{x'_2}^{\mathbf{q}}D_{y'_2}^{\mathbf{q}}O_{z'_2}^{\mathbf{q}}$
is
\begin{equation}
    r_{C_{x_1}^{\mathbf{p}}D_{y_1}^{\mathbf{p}}O_{z_1}^{\mathbf{p}}C_{x_2}^{\mathbf{q}}D_{y_2}^{\mathbf{q}}O_{z_2}^{\mathbf{q}} \rightarrow C_{x'_1}^{\mathbf{p}}D_{y'_1}^{\mathbf{p}}O_{z'_1}^{\mathbf{p}}C_{x'_2}^{\mathbf{q}}D_{y'_2}^{\mathbf{q}}O_{z'_2}^{\mathbf{q}}} = r_{C_{x_1}^{\mathbf{p}}D_{y_1}^{\mathbf{p}}O_{z_1}^{\mathbf{p}}\rightarrow C_{x'_1}^{\mathbf{p}}D_{y'_1}^{\mathbf{p}}O_{z'_1}^{\mathbf{p}}} \cdot r_{C_{x_2}^{\mathbf{q}}D_{y_2}^{\mathbf{q}}O_{z_2}^{\mathbf{q}}\rightarrow C_{x'_2}^{\mathbf{q}}D_{y'_2}^{\mathbf{q}}O_{z'_2}^{\mathbf{q}}}.
    \label{equ:reduced_transition}
\end{equation}

In the meanwhile, the initial probability to be in one of the $(k+1)(k+2)(n-k+1)(n-k+2)/4$ states in the very first round is given by
\begin{equation}
    \begin{aligned}
        u_{C_{x_1}^{\mathbf{p}}D_{y_1}^{\mathbf{p}}O_{z_1}^{\mathbf{p}}C_{x_2}^{\mathbf{q}}D_{y_2}^{\mathbf{q}}O_{z_2}^{\mathbf{q}}}(0)=&\tbinom{k}{x_1}\tbinom{k-x_1}{y_1}\left(p_0^C\right)^{x_1}\left(p_0^D\right)^{y_1}\left(p_0^O\right)^{z_1}\\
        &\tbinom{n-k}{x_2}\tbinom{n-k-x_2}{y_1}\left(q_0^C\right)^{x_2}\left(q_0^D\right)^{y_2}\left(q_0^O\right)^{z_2}.
    \end{aligned}
    \label{equ:reduced_initial}
\end{equation}
To calculate the long-term payoffs, we collect all probabilities in Eq.~\eqref{equ:reduced_transition} and build a transition matrix $\mathbf{R}$. Similarly, we collect all probabilities in Eq.~\eqref{equ:reduced_initial} and build a row-vector $\mathbf{u}(0)$. Moreover, let $\mathbf{u}(t)$ be the row-vector that represents the distribution over all possible states $C_{x_1}^{\mathbf{p}}D_{y_1}^{\mathbf{p}}O_{z_1}^{\mathbf{p}}C_{x_2}^{\mathbf{q}}D_{y_2}^{\mathbf{q}}O_{z_2}^{\mathbf{q}}$ in round $t$.

By replacing $\mathbf{M}$ and $\mathbf{v}$ in Eq.~(\ref{equ:average_distribution_noDiscounting}) and Eq.~(\ref{equ:average_distribution}) with $\mathbf{R}$ and $\mathbf{u}$, respectively, we obtain a distribution $\mathbf{u}$ for payoff calculations,
where its entry $u_{C_{x_1}^{\mathbf{p}}D_{y_1}^{\mathbf{p}}O_{z_1}^{\mathbf{p}}C_{x_2}^{\mathbf{q}}D_{y_2}^{\mathbf{q}}O_{z_2}^{\mathbf{q}}}$ represents the probability that one finds itself in state $C_{x_1}^{\mathbf{p}}D_{y_1}^{\mathbf{p}}O_{z_1}^{\mathbf{p}}C_{x_2}^{\mathbf{q}}D_{y_2}^{\mathbf{q}}O_{z_2}^{\mathbf{q}}$ over the course of the play. Thus, the long-term payoff of $\mathbf{p}$ players is
\begin{equation}
    \begin{aligned}
    \pi_{\mathbf{p}}=\frac{1}{k}&\sum_{x_1,y_1,z_1,x_2,y_2,z_2}u_{C_{x_1}^{\mathbf{p}}D_{y_1}^{\mathbf{p}}O_{z_1}^{\mathbf{p}}C_{x_2}^{\mathbf{q}}D_{y_2}^{\mathbf{q}}O_{z_2}^{\mathbf{q}}}\left[x_1H(n-z_1-z_2-2)[\frac{(x_1+x_2)rc}{n-z_1-z_2}-c]\right.\\
    &\left.+y_1H(n-z_1-z_2-2)\frac{(x_1+x_2)rc}{n-z_1-z_2}+(z_1 +H(2+z_1+z_2-n))\sigma\right],
    \end{aligned}
    \label{equ:longterm_payoff_p}
\end{equation}
where
\begin{equation}
    H(x)=\begin{cases}
        1, & x\geq0 \\
        0, & x<0
    \end{cases}
\end{equation}
is Heaviside step function. And the long-term payoff of $\mathbf{q}$ players is
\begin{equation}
    \begin{aligned}
        \pi_{\mathbf{q}}=\frac{1}{n-k}&\sum_{x_1,y_1,z_1,x_2,y_2,z_2}u_{C_{x_1}^{\mathbf{p}}D_{y_1}^{\mathbf{p}}O_{z_1}^{\mathbf{p}}C_{x_2}^{\mathbf{q}}D_{y_2}^{\mathbf{q}}O_{z_2}^{\mathbf{q}}}\left[x_2H(n-z_1-z_2-2)[\frac{(x_1+x_2)rc}{n-z_1-z_2}-c]\right.\\
        &\left.+y_2H(n-z_1-z_2-2)\frac{(x_1+x_2)rc}{n-z_1-z_2}+(z_2 +H(2+z_1+z_2-n))\sigma\right].
    \end{aligned}
    \label{equ:longterm_payoff_q}
\end{equation}

\subsection{Cooperation rate, defection rate and opt-out rate}\label{sec:calculaterate}
To calculate the cooperation rate, defection rate and opt-out rate of a reactive strategy $\mathbf{p}$, we set $k=n$ in the above derivations, i.e., all individuals adopting the same strategy $\mathbf{p}$. Then, the state becomes $C_{x_1}^{\mathbf{p}}D_{y_1}^{\mathbf{p}}O_{z_1}^{\mathbf{p}}$ and
the cooperation rate of strategy $\mathbf{p}$ is given by
\begin{equation}
    \gamma_C=\frac{1}{n}\sum_{x_1,y_1,z_1}x_1H(n-z_1-2)u_{C_{x_1}^{\mathbf{p}}D_{y_1}^{\mathbf{p}}O_{z_1}^{\mathbf{p}}},
    \label{equ:cooperationrate}
\end{equation}
the defection rate
\begin{equation}
    \gamma_D=\frac{1}{n}\sum_{x_1,y_1,z_1}y_1H(n-z_1-2)u_{C_{x_1}^{\mathbf{p}}D_{y_1}^{\mathbf{p}}O_{z_1}^{\mathbf{p}}},
    \label{equ:defectionrate}
\end{equation}
and the opt-out rate
\begin{equation}
    \gamma_O=\frac{1}{n}\sum_{x_1,y_1,z_1}(z_1+H(n-z_1-2))u_{C_{x_1}^{\mathbf{p}}D_{y_1}^{\mathbf{p}}O_{z_1}^{\mathbf{p}}}.
    \label{equ:optoutrate}
\end{equation}

By the state-clustering method, we only need to handle a $\frac{(n+1)(n+2)}{2}$-dimensional matrix rather than a $3^n$-dimensional matrix when calculating the cooperation rate, defection rate and opt-out rate. This reduces the computation time significantly.

\subsection{Evolutionary dynamics}\label{sec:evolutionarydyanmics}
On a longer time scale, we assume that individuals can change their strategies. Here, we focus on the pairwise comparison process where individuals either explore randomly or imitate others to adopt more profitable strategies.

Specifically, we consider a well-mixed population of $N$ individuals. Each individual is equipped with a reactive strategy. In each evolutionary step, $n$ individuals are randomly selected to engage in a repeated optional PGG. The payoff that each individual obtains in the repeated optional PGG is calculated according to Eq.~\eqref{equ:long_term_payoff} (or according to Eqs.~\eqref{equ:longterm_payoff_p} and \eqref{equ:longterm_payoff_q} if there are two different strategies in the population).
Given the strategy configuration of the population, we compute the expected payoff $\pi$ for each individual. Then, it comes to the stage of strategy updating. At each time step, an individual $i$ is randomly drawn from the population to update its strategy. This individual either explores by adopting a random reactive strategy (corresponding to mutations) with probability $\mu$ or implements imitation (corresponding to natural selection) with the rest probability $1-\mu$.
Denote the expected payoff of individual $i$ and the role model $j$ as $\pi_i$ and $\pi_j$, respectively.
If individual $i$ imitates, it randomly selects a role model $j$ ($j\neq i$) and adopts its strategy with the probability
\begin{equation}
    \phi_{ij}=\frac{1}{1+\exp[-s(\pi_j-\pi_i)]},
\end{equation}
where $s$ is called the selection intensity, quantifying the contribution of expected payoffs to strategy imitation. When $s=0$, it is the neutral drift and individual $i$ imitates $j$'s strategy with probability $1/2$, independent of their payoffs. Under strong selection, i.e., $s\rightarrow\infty$, imitation occurs only when the role model's payoff is no less than that of individual $i$; otherwise, individual $i$ keeps its strategy unchanged.

The above evolutionary process can be modeled by a Markov chain. Here, we assume that mutations are so rare that before another mutant occurs, the current mutant either takes over the population or is wiped out. In this limiting case ($\mu\rightarrow 0$), there are at most two different strategies in the population and the population stays most of its time in the homogeneous states where the population consists of only one strategy. Let us denote the mutant strategy as $M$ and the resident strategy as $R$. Note that from the perspective of strategies, the above pairwise comparison process for individuals' strategy adaptation is also a mutation-selection process for strategies, where more successful strategies are more likely to spread. In this sense, we calculate the probability that a mutant strategy successfully takes over the resident population, i.e., the fixation probability. Under the limit of rare mutations ($\mu\rightarrow 0$), the fixation probability for strategy $M$ to take over strategy $R$ is
\begin{equation}
    \rho_{R\rightarrow M}=\frac{1}{1+\sum_{i=1}^{N-1}\prod_{j=1}^i \exp\left[-s(\pi_M(j)-\pi_R(j))\right]},
    \label{equ:fixation_probability}
\end{equation}
where $\pi_M(j)$ and $\pi_R(j)$ are the expected payoff of the mutant and resident when there are $j$ mutants in the population.

For numerical investigations, we simulate the evolutionary dynamics of strategies by the method proposed in~\cite{imhof2010stochastic}. Specifically, we initialize the population as a homogeneous population where all individuals adopt a reactive strategy. Then, we introduce a randomly selected mutant to the population. The mutant either takes over the population with probability in Eq.~\eqref{equ:fixation_probability} or is wiped out with the rest probability. No matter what happens, the population becomes homogeneous again. At this time, we introduce another mutant randomly drawn from the set of reactive strategies and iterate the above process. Eventually, the evolutionary dynamics visit all possible strategies (i.e., homogeneous population states) if the simulation iterates sufficiently many times.
During the evolutionary process, the cooperation rate, defection rate and opt-out rate at each time step are calculated by applying Eq.~\eqref{equ:cooperationrate}, Eq.~\eqref{equ:defectionrate} and Eq.~\eqref{equ:optoutrate} to the very strategy that makes up the homogeneous population at that time. Moreover, the most abundant strategy is the one that occupies the population for the longest time during evolution.

\section{Equilibrium analysis for repeated optional PGGs}\label{sec:equilibriumanalysis}
In this section, we employ equilibrium analysis to characterize all Nash equilibria (more accurately, subgame perfect equilibria, SPE) in the space of reactive strategies. To find all Nash equilibria, we first prove that if all individuals adopt the same reactive strategy, the game dynamics of a repeated optional PGG can have only eight possible endings, including persistent cooperation, persistent defection, and so on. We then identify all reactive strategies that support persistent cooperation and form Nash equilibria. Among these strategies, we find three novel classes of strategies, $OD$-$TFT$, $DO$-$OFT$ and $OD$-$TFT_C$, which are robust to implementation errors and become a Nash equilibrium with the lowest threshold for the multiplication factor to surpass. Furthermore, we prove that there are no Nash equilibria giving rise to persistent defection. This implies that repeated optional PGGs do not favor defection, which sharply contrasts with repeated compulsory PGGs. Finally, we also characterize all Nash equilibria that support persistent opt-out.

\subsection{Possible endings of game dynamics}
In the following, we first analyze all possible endings of the game dynamics that a repeated optional PGG has if all individuals adopt the same pure reactive strategy. We find that there are at most eight possible endings.
\begin{theorem}[Possible endings of game dynamics in a group]
    In a repeated optional public goods game, if all group members adopt the same pure reactive strategy and always perfectly implement their actions, the corresponding game dynamics have at most eight possible endings:
    \begin{itemize}
        \item [\textbf{1},] persistent (full) cooperation;
        \item [\textbf{2},] persistent (full) defection;
        \item [\textbf{3},] persistent (full) opt-out;
        \item [\textbf{4},] alternating between full cooperation and full defection;
        \item [\textbf{5},] alternating between full cooperation and full opt-out;
        \item [\textbf{6},] alternating between full defection and full opt-out;
        \item [\textbf{7},] cycling from full cooperation to full defection, to full opt-out, and back to full cooperation;
        \item [\textbf{8},] cycling from full cooperation to full opt-out, to full defection, and back to full cooperation.
    \end{itemize}
    \label{lemma:dynamicsofgroup}
\end{theorem}
\begin{proof}
    When all group members adopt the same pure strategy and always implement their actions perfectly (i.e., without implementation errors), they take the same action in any given round. That is, in such a group, the only possible game states are $(n,0)$, $(0,n)$, and $(0,0)$. In the meanwhile, for any game state, there is one and only one game state that it can transit into, given that all individuals use the same pure strategy. Enumerating all possible transitions between these game states, we obtain eight possible group configurations, as shown in Supplementary Fig.~\ref{Sfig:8dynamics}.
\end{proof}

\subsection{Finite state automaton}
To search for Nash equilibria that support a specific ending of the game dynamics, we introduce a useful representation of pure strategies, the finite state automaton~\cite{Nowak1995,Hilbe2018b}. The finite state automaton of a strategy $\mathbf{p}$ consists of three parts: (1) the set of state, $\mathcal{S}_{\mathbf{p}}$; (2) the set of possible game state in the previous round, $\mathcal{H}=\{(x,y)\}$; and (3) the transition function $Q_{\mathbf{p}}:\mathcal{S}_{\mathbf{p}}\times\mathcal{H}\rightarrow\mathcal{S}_\mathbf{p}$. Each state, $s\in\mathcal{S}$, is associated with an action, $A_s$ ($A_s\in\{C,D,O\}$): when an individual is in the state $s$, it takes action $A_s$ in the current round. The transition function specifies the state that the individual is in after the corresponding game state, $h\in\mathcal{H}$. We can calculate the value of $Q_{\mathbf{p}}(s,h)$ according to $\mathbf{p}$ and the action associated with each state. The cardinality of the set $\mathbf{S}_{\mathbf{p}}$ measures the complexity of $\mathbf{p}$. A finite state automaton with more states means that the corresponding strategy is more complex. We find that the complexity of a pure strategy is at most three.

\begin{lemma}[Complexity of reactive strategies]
    In a repeated optional public goods game, the finite state automaton of a pure reactive strategy has at most three states, State $0$, State $1$ and State $2$, which are associated with cooperation ($C$), defection ($D$) and opt-out ($O$), respectively.
    \label{lemma:finite_complexity}
\end{lemma}
\begin{proof}
    To prove Lemma \ref{lemma:finite_complexity}, we show that a pure reactive strategy $\mathbf{p}=(p_h)$ can be represented as the finite state automaton with the set of states, $\mathcal{S}_{\mathbf{p}}=\{0,1,2\}$, the set of possible game states, $\mathcal{H}=\{(x,y)|0\leq x+y \leq n\}$, and the transition function $Q(s,h)$, where
\begin{equation*}
    Q(s,h)=
    \begin{cases}
        0, & h\in\mathcal{H}_C \text{ and } s\in\mathcal{S}_{h} \\
        1, & h\in\mathcal{H}_D \text{ and } s\in\mathcal{S}_h \\
        2, & h\in\mathcal{H}_O \text{ and } s\in\mathcal{S}_h
    \end{cases}.
\end{equation*}
Here, the individual in State $0$, State $1$ and State $2$ cooperates, defects and opts out of the game in the current round, respectively. $\mathcal{H}_C=\{(x,y)|p_{x,y}^C=1\}$, $\mathcal{H}_D=\{(x,y)|p_{x,y}^D=1\}$, and $\mathcal{H}_O=\{(x,y)|p_{x,y}^O=1\}$ represent the set of game states after which an individual with $\mathbf{p}$ cooperates, defects and opts out in the next round, respectively. Because strategy $\mathbf{p}$ is pure, $\mathcal{H}_C\cup\mathcal{H}_D\cup\mathcal{H}_O=\mathcal{H}$. The set $\mathcal{S}_h$ is the set of all feasible states. Let $\mathcal{S}_{\text{illegal}}=\{s|s=0\text{ and }h=(0,y)\}\cup\{s|s=1\text{ and }h=(x,0)\}\cup\{s|s=2,h=(x,y),\text{ and } x+y=n\}$ denote the illegal combinations between states of the automaton $s$ and game states $h$. The set of all feasible states can be written as $\mathcal{S}_h=\mathcal{S}_{\mathbf{p}}\backslash\mathcal{S}_{\text{illegal}}$.

We now show that the sequence of actions that an individual acts according to the above finite automaton is the same as that according to strategy $\mathbf{p}$. By the definitions of $\mathcal{H}_C$, $\mathcal{H}_D$ and $\mathcal{H}_O$, an individual with $\mathbf{p}$ cooperates after $h\in\mathcal{H}_C$, defects after $h\in\mathcal{H}_D$, and opts out of the game after $h\in\mathcal{H}_O$. For an individual that acts according to the above finite state automaton, it moves to State $0$ after $\mathcal{H}_C$, to State $1$ after $\mathcal{H}_D$ and to State $2$ after $\mathcal{H}_O$. Since State $0$ is associated with cooperation, State $1$ with defection and State $2$ with opt-out, the individual acts the same as it does according to $\mathbf{p}$. Thus, $\mathbf{p}$ can be represented as the above finite state automaton, which has three states.
\end{proof}

The finite state automaton of a pure reactive strategy is finite since an individual determines its action based on the game state in the previous round. The finite state automaton contains at most three states since each individual has three options in each round, and each state corresponds to one action.

With the finite state automaton, we can describe the sequence of an individual's actions. Moreover, Lemma \ref{lemma:finite_complexity} is useful in checking whether a pure reactive strategy is an equilibrium. In the following sections, we make use of Lemma \ref{lemma:finite_complexity} to find all pure reactive strategies that give rise to specific endings of the game dynamics and form an equilibrium.

\subsection{Equilibrium in repeated optional PGGs}
In this section, we identify all pure reactive strategies that support persistent cooperation, persistent defection and persistent opt-out.
\subsubsection{Subgame perfect equilibria that support persistent cooperation}\label{sec:necooperation}
We start by identifying all pure reactive strategies that support persistent (full) cooperation (ending \emph{\textbf{1}}) and form an equilibrium. By analyzing the properties of these strategies and the conditions that they become equilibria, we may explain why repeated optional PGGs support almost full cooperation when repeated compulsory PGGs hardly promote cooperation. By resorting to Lemma \ref{lemma:finite_complexity}, we have the following theorem.
\begin{theorem}[Pure reactive strategies that support persistent cooperation and form an equilibrium]
   Consider a repeated optional public goods game with $1<r<n$ and $0<\sigma<(r-1)c$. Let $\mathbf{p}$ be a pure reactive strategy that cooperates in the first round and sticks to cooperation after all other individuals do so (i.e. $p_{n,0}^C=1$). Then, $\mathbf{p}$ is a subgame perfect equilibrium if and only if it is one of the strategies shown in Supplementary Fig.~\ref{Sfig:feasiblerangecooperation}\textbf{a} with the corresponding condition held.
    \label{thm:cooperation}
\end{theorem}
The proofs of Theorem \ref{thm:cooperation} and all the subsequent theorems are presented in the appendix of this Supplementary Information. 

Theorem \ref{thm:cooperation} considers the case where all individuals perform the action perfectly. If individuals are subject to implementation errors, the strategies from the first to fifth row and the eleventh row in Supplementary Fig.~\ref{Sfig:feasiblerangecooperation}\textbf{a} no longer support persistent cooperation. We take the strategy in the first row as an example to explain why it is sensitive to implementation errors in supporting persistent cooperation. Suppose all individuals adopt one and the same strategy of the first row in Supplementary Fig.~\ref{Sfig:feasiblerangecooperation}\textbf{a} and cooperate in the previous round. If each of them executes its strategy perfectly, they stick to cooperation in all rounds. Once someone takes action $D$ mistakenly, they move to game state $(n-1,1)$. Without loss of generality, we denote the round that individuals are in game state $(n-1,1)$ as round $t$. Subsequently, individuals move to game state $(0,n)$ in round $t+1$, to game state $(0,0)$ in round $t+2$ and so on. Individuals switch between game states $(0,n)$ and $(0,0)$ thereafter unless another implementation error occurs. Thus, the strategies in the first row of Supplementary Fig.~\ref{Sfig:feasiblerangecooperation}\textbf{a} do not support persistent cooperation. Similarly, we have that the strategies from the second to fifth row and the eleventh row are also prone to implementation errors. In the meanwhile, we have that all other strategies listed in Supplementary Fig.~\ref{Sfig:feasiblerangecooperation}\textbf{a} can robustly support persistent cooperation. For example, for a group of individuals adopting one of the strategies from the sixth row in Supplementary Fig.~\ref{Sfig:feasiblerangecooperation}\textbf{a}, individuals defect if they move to game state $(n-1,1)$ accidentally; they subsequently opt out and restore cooperation after game state $(0,0)$. If individuals move to game state $(x,y)$ else by mistake, they move to either game state $(n,0)$, $(0,n)$ or $(0,0)$ in the next round and restore cooperation within three rounds, thus supporting persistent cooperation.

The strategy $AoN$, which gives rise to persistent cooperation in repeated compulsory PGGs~\cite{Pinheiro2014,Hilbe2017}, is a special case of the strategies in the ninth row in Supplementary Fig.~\ref{Sfig:feasiblerangecooperation}\textbf{a}. It becomes an equilibrium if
\begin{equation}
    r\geq\frac{2n}{n+1}.
\end{equation}
We plot the feasible region of the multiplication factor $r$ and the payoff for opt-out $\sigma$ for each strategy in Supplementary Fig.~\ref{Sfig:feasiblerangecooperation}\textbf{a} to form an equilibrium, as shown in Supplementary Fig.~\ref{Sfig:feasiblerangecooperation}\textbf{b}. Surprisingly, we find three classes of strategies, namely $OD$-$TFT$, $DO$-$OFT$ and $OD$-$TFT_C$ (the sixth to eighth rows in Supplementary Fig.~\ref{Sfig:feasiblerangecooperation}\textbf{a}), which become equilibria with a lower threshold for the multiplication factor $r$ than $AoN$. The region where $OD$-$TFT$, $DO$-$OFT$ and $OD$-$TFT_C$ are equilibrium while $AoN$ is not is
\begin{equation}
    \frac{3n}{2n+1} \leq r < \frac{2n}{n+1} \text{ and } \sigma \leq \frac{2n+1}{n}rc - 3c.
\end{equation}
Evolutionary analysis shows that the dynamic process visits $OD$-$TFT$, $DO$-$OFT$ and $OD$-$TFT_C$ more than $96\%$ of the time, though the probability that a randomly selected pure reactive strategy belongs to $OD$-$TFT$, $DO$-$OFT$ and $OD$-$TFT_C$ is only $0.78\%$. Combined with  ``knock-out experiments'' as shown in Fig. 4 in the main text, we find that $OD$-$TFT$, $DO$-$OFT$ and $OD$-$TFT_C$ contribute to the high level of cooperation in repeated optional PGGs. The advantage of $OD$-$TFT$, $DO$-$OFT$ and $OD$-$TFT_C$ is shown in the following corollary.
\begin{corollary}
    In a repeated optional public goods game with $1<r<n$ and $0<\sigma<(r-1)c$, $OD$-$TFT$, $DO$-$OFT$ and $OD$-$TFT_C$ become subgame perfect equilibria while $AoN$ does not if
    \begin{equation}
        \frac{3n}{2n+1} \leq r < \frac{2n}{n+1} \text{ and } \sigma \leq \frac{2n+1}{n}rc - 3c.
    \end{equation}
\end{corollary}

Indeed, $OD$-$TFT$, $DO$-$OFT$ and $OD$-$TFT_C$ have similar characteristics to memory-two strategies in repeated compulsory PGGs. An individual with a memory-two strategy conditions its actions on game states in the previous two rounds. Thus, the individual can retaliate against defection for two rounds, e.g., by adopting $AoN_2$ strategy~\cite{Hilbe2017}. Retaliation for two rounds further reduces the incentive to deviate from full cooperation, and thus makes individuals more hesitant to deviate. The memory-one strategies in repeated optional PGGs, $OD$-$TFT$, $DO$-$OFT$ and $OD$-$TFT_C$, enforce similar retaliation by ordered actions of defection and opt-out. This similarity suggests that the advantage of a long memory strategy may be realized by extending the set of available actions. In section \ref{sec:two_additional_options}, we further explore the evolution of cooperation in repeated PGGs with two kinds of actions for opting out and find that the level of cooperation increases further.

\subsubsection{Subgame perfect equilibria that support persistent defection}
In this section, we aim to search for all pure reactive strategies that give rise to persistent defection (ending \emph{\textbf{2}}) and form equilibria. The following theorem claims that there is no equilibrium supporting persistent defection.
\begin{theorem}[Existence of pure reactive strategies that favor persistent defection and form equilibria] In a repeated optional game with $1<r<n$ and $0<\sigma<(r-1)c$, no pure reactive strategy gives rise to persistent defection and form an equilibrium.
\label{thm:defection}
\end{theorem}
The above theorem suggests that a population will not fall into the deadlock of mutual defection in repeated optional PGGs. The theorem can be interpreted intuitively as follows. If all individuals adopt a strategy $\mathbf{p}$ that supports persistent defection, they get the long-term payoff of zero. Compared with repeated compulsory PGGs, repeated optional PGGs provide a better option than defection, namely opt-out, which provides individuals a higher payoff of $\sigma>0$. Thus, individuals have an incentive to deviate.

\subsubsection{Subgame perfect equilibria that support persistent opt-out}
In this section, we identify all pure reactive strategies that support persistent opt-out and form equilibria. The following theorem identifies all equilibria that give rise to persistent opt-out.
\begin{theorem}[Pure reactive strategies that support persistent opt-out and form equilibria]
    Consider a repeated optional public goods game with $1<r<n$ and $0<\sigma<(r-1)c$. Let $\mathbf{p}$ be a pure reactive strategy that opts out in the first round and adheres to opt-out once all other individuals also opt out (i.e. $p_{0,0}^O=1$). Then, $\mathbf{p}$ is a subgame perfect equilibrium if and only if it belongs to the set of strategies listed in Supplementary Fig.~\ref{Sfig:feasiblerangeoptout}\textbf{a} with the corresponding condition held.
    \label{thm:optout}
\end{theorem}
\begin{remark}
    The above theorem considers the strategies that opt out in the first round. For the strategies in the second row of Supplementary Fig.~\ref{Sfig:feasiblerangeoptout}\textbf{a}, even if it defects in the first round, the ending of game dynamics is still persistent opt-out. The strategies in the third, fourth, and fifth rows of Supplementary Fig.~\ref{Sfig:feasiblerangeoptout}\textbf{a} still support persistent opt-out no matter whether they cooperate or defect in the very first round.
\end{remark}
\begin{remark}
    Except for the strategy in the first and second rows of Supplementary Fig.~\ref{Sfig:feasiblerangeoptout}\textbf{a}, other strategies are robust to errors.
\end{remark}

\section{Repeated optional prisoner's dilemmas}
In the main text, we have shown that voluntary participation promotes the evolution of cooperation in repeated optional PGGs. To better compare our results with that in classical pairwise games, in the following, we apply our analysis in repeated optional prisoner's dilemmas.

In each round of repeated optional prisoner's dilemmas, if both individuals cooperate, they each get a payoff $R$. If both individuals defect, each individual receives a payoff $P$. If one individual cooperates and the other defects, the individual who cooperates gets the lowest payoff $S$ and the individual who defects obtains the highest payoff $T$. If there is an individual opting out, the game is canceled and each individual gets a payoff $\sigma$. It is assumed that mutual cooperation yields a higher payoff than mutual defection and that opt-out leads to a payoff less than mutual cooperation but higher than mutual defection (i.e. $T>R>\sigma>P>S$). In what follows, we consider the payoffs $T=b'$, $R=b'-c'$, $P=0$ and $S=-c'$, which is so-called optional ``donation games'' (DGs). Thus, the payoff matrix in an optional DG can be represented as
\begin{equation}
    \bordermatrix{& C & D & O \cr
    C & b'-c' & -c' & \sigma \cr
    D & b' & 0 & \sigma \cr
    O & \sigma & \sigma & \sigma}.
\end{equation}
In an optional PGG with two individuals, the payoff matrix is written as
\begin{equation}
    \bordermatrix{& C & D & O \cr
    C & (r-1)c & rc/2-c & \sigma \cr
    D & rc/2 & 0 & \sigma \cr
    O & \sigma & \sigma & \sigma}.
\end{equation}
Thus, an optional DG is equivalent to an optional PGG if and only if
\begin{equation}
    \begin{aligned}
        b'&=rc/2 \\
        c'&=c-rc/2
    \end{aligned}.
    \label{equ:relationshipbr}
\end{equation}

We simulate the evolutionary dynamics of repeated optional DGs, the corresponding repeated compulsory DGs and one-shot optional DGs. Similar to evolutionary results in repeated optional PGGs, repeated optional DGs boost the evolution of cooperation markedly even if the corresponding repeated compulsory DGs and one-shot optional DGs fail to (Supplementary Fig.~\ref{Sfig:donationGame}\textbf{a}).
$OD$-$TFT$ and $DO$-$OFT$ play an important role in the evolution of cooperation in repeated optional DGs (Supplementary Fig.~\ref{Sfig:donationGame}\textbf{c,d}). Evolutionary dynamics visit $OD$-$TFT$ and $DO$-$OFT$ about $98.24\%$ of time although the probability that a randomly selected strategy belongs to $OD$-$TFT$ and $DO$-$OFT$ is less than $1\%$. In repeated optional DGs, $OD$-$TFT$ and $DO$-$OFT$ are represented as shown in Supplementary Fig.~\ref{Sfig:donationGame}\textbf{b}. Substituting Eq.~\eqref{equ:relationshipbr} into Supplementary Fig.~\ref{Sfig:feasiblerangecooperation}\textbf{a}, it yields that $OD$-$TFT$ and $DO$-$OFT$ form an equilibrium if and only if
\begin{equation}
    \frac{b'}{c'}\geq\frac{\sigma}{2c'}+\frac{3}{2}.
\end{equation}
Note that $WSLS$, which is critical to the evolution of cooperation in repeated compulsory DGs, becomes Nash equilibrium if $b'/c'\geq 2$. Therefore, if $\sigma<c'$, $OD$-$TFT$ and $DO$-$OFT$ have a lower threshold to surpass to become a Nash equilibrium. In this case, repeated optional DGs are more likely to promote the evolution of cooperation. The class of strategies, $OD$-$TFT_C$, does not exist in repeated optional DGs since $OD$-$TFT_C$ needs to prescribe two different actions after game state $(1,1)$ (needs to prescribe $D$ after game state $(n-1,1)=(1,1)$ and $C$ after game state $(1,n-1)=(1,1)$).

\section{Further extensions}\label{sec:further}
In what follows, we extend our basic model by considering two scenarios: (1) individuals opt out of the game with a risk of failure (Section \ref{sec:opting_out_failure}); and (2) individuals can choose two kinds of actions to implement opt-out (Section \ref{sec:two_additional_options}). In addition, we make a detailed analysis of the evolution of cooperation in one-shot optional PGGs in Section \ref{sec:connection_to_oneshot}.

\subsection{Repeated optional PGGs with failures of opting out}
\label{sec:opting_out_failure}
In the previous sections, we explore the evolution of cooperation in repeated optional PGGs, where individuals can freely opt out of the interaction. In this section, we extend our basic model to the scenario where individuals may fail to opt out of the game. This scenario captures a delayed effect of opting out or restrictions on the use of opt-out.

In a repeated optional PGG with failure of opting out, the individual who intends to opt out fails to do so with probability $\alpha\in[0,1]$. When $\alpha=0$, individuals can timely opt out of the interaction. The model reduces to the basic model discussed in previous sections. When $\alpha=1$, individuals are forbidden to opt out. The model reduces to repeated compulsory PGGs. When $\alpha\in(0,1)$, there is always a positive probability of failure for opting out. We assume that if an individual intends to opt out but fails, it keeps implementing its previous action.

In a repeated optional PGG with failure probability $\alpha$, the effective strategy can be calculated explicitly. For an individual with a reactive strategy $\mathbf{p}=(p_{x,y}^C,p_{x,y}^D)$, its effective strategy becomes
\begin{equation}
    \mathbf{p}^{\alpha}=
    \begin{cases}
    (1-\alpha)p_{x,y}^O & \text{ if } A_i\in\{C,D\} \\
    p_{x,y}^C+\alpha p_{x,y}^O & \text{ if } A_i=C \\
    p_{x,y}^D+\alpha p_{x,y}^O & \text{ if } A_i=D \\
    p_{x,y}^A & \text{ otherwise }
    \end{cases},
\end{equation}
where $A_i$ is the action of individual $i$ in the previous round. Note that the failure of opting out happens after an individual takes action $O$. Thus, the action of an individual is influenced firstly by the ``trembling hands'', then by the probability of failure. Suppose that all individuals adopt reactive strategies. By the effective strategy, we can calculate the long-term payoffs by the state-clustering method (described in Section \ref{sec:stateclustering}).

Since individuals may fail to opt out timely, it seems that the effect of opt-out on the evolution of cooperation will be reduced, especially when the probability of failure is high. Counterintuitively, our simulation shows that the evolution of cooperation is robust to a wide range of failure probability, $\alpha$. Only in the games with a failure probability greater than $0.99$, the level of cooperation is lower than $50\%$ (Fig. 5\textbf{b} in the main text).

To figure out the underlying mechanism that helps cooperation evolve, we calculate the average effective strategy that dominates in evolutionary dynamics under a high failure probability ($\alpha=0.9$), as shown in Fig. 5\textbf{c} in the main text. It turns out that the average effective strategy behaves similarly to $OD$-$TFT$: it cooperates when all individuals cooperate; it defects when one individual defects and other individuals cooperate; it intends to opt out when all individuals defect; and it cooperates when all individuals opt out. Compared with $OD$-$TFT$, the strategy evolves additional properties to ensure that individuals restore cooperation quickly (that is, move to the game state $(0,0)$).

To better illustrate the properties, we plot the average effective strategy of the top 50 strategies (Supplementary Fig.~\ref{Sfig:delayed}\textbf{a}) and the average effective strategy of strategies ranked 51 to 100 (Supplementary Fig.~\ref{Sfig:delayed}\textbf{b}) in the evolutionary dynamics of repeated three-player optional PGGs. In what follows, we explain how the dominant strategies sustain cooperation. We first consider the top 50 strategies. Let $\mathcal{P}_{50}$ denote the set of the top 50 strategies. Suppose that all individuals adopt the same strategy in set $\mathcal{P}_{50}$ and cooperate in the current round. If they stick to their strategies, all individuals keep cooperating. Once an individual defects, all individuals defect and the game state moves to $(D,DD)$. Here, $(A^*,AA')$ represents the game state that includes all individuals' actions, where $A^*$ is the action of the focal individual and $AA'$ are the actions of the other two individuals. After game state $(D, DD)$, all individuals intend to opt out of the game. When one individual $i^*$ succeeds in opting out, the game state moves to $(O, DD)$ from the perspective of individual $i^*$ and to $(D, DO)$ from the perspective of other individuals. Then, individual $i^*$ continues opting out after game state $(O, DD)$ and in the meanwhile, the other two individuals continue trying to opt out of the interaction. If one of them succeeds, the game state moves to $(O, DO)$ from the perspective of the individual who succeeds and to $(D, OO)$ from the perspective of the individual who fails. Subsequently, all individuals cooperate and the game state moves back to mutual cooperation.

We next consider the strategies ranked 51 to 100. Let $\mathcal{P}_{51-100}$ denote the strategies ranked 51 to 100. Suppose that all individuals adopt the same strategy in set $\mathcal{P}_{50-100}$ and cooperate in the current round. If an individual defects, all individuals defect in the next round and try to opt out after game state $(D,DD)$. Once someone opts out successfully, individuals move to game state  $(O, DD)$ from the perspective of the individual who succeeds in opting out and to $(D, DO)$ from the perspective of other individuals. While in these states, individuals cooperation in the next round and restore cooperation directly. When individuals adopt strategies in sets $\mathcal{P}_{50}$ and $\mathcal{P}_{50-100}$, it is unnecessary for all individuals to opt out in the same round, thus restoring cooperation quickly.

\subsection{Repeated PGGs with two kinds of actions for opting out}\label{sec:two_additional_options}
In section \ref{sec:necooperation}, we find that $OD$-$TFT$ and $DO$-$OFT$, which play a crucial role in promoting the evolution of cooperation in repeated optional PGGs, possess similar characteristics as $AoN_2$. The additional option, opt-out, together with defection, can help record the number of rounds that individuals have not cooperated since an individual deviates from full cooperation. This leads to an interesting question: can repeated PGGs with more options to opt out yield a higher level of cooperation?

To answer this question, we consider repeated PGGs with two kinds of actions for opting out. In each round, individuals have four available actions, $C$ (cooperation), $D$ (defection), $O_1$ (opt-out), and $O_2$ (opt-out). The individual who decides to cooperate or defect will participate in the PGG and obtains the payoff as described in section \ref{sec:gamesetup}. Once an individual decides to take action $O_1$ or $O_2$, it opts out of the interaction and gets a fixed payoff $\sigma$. Although $O_1$ and $O_2$ bring an equal payoff, they are treated as different actions. This can be viewed as the case that individuals have two different ways to opt out and obtain $\sigma$.

The repeated games last for infinitely many rounds. In such repeated games, individuals determine their actions based on action profiles in the previous round. Here, we assume that individuals take into account the number of each action to make decisions. Let $x$, $y$ and $z$ denote the number of individuals that take action $C$, $D$, and $O_1$, respectively. The game state in each round can be written as $(x,y,z)$ and the strategy can be written as a $(n+1)(n+2)(n+3)/2$-dimensional vector
\begin{equation*}
    \mathbf{p}=(p_{x,y,z}^C,p_{x,y,z}^D,p_{x,y,z}^{O_1}),
\end{equation*}
where $p_{x,y,z}^A$ with $0\leq p_{x,y,z}^C+p_{x,y,z}^D+p_{x,y,z}^{O_1}\leq 1$ represents the probability that an individual with $\mathbf{p}$ takes action $A\in\{C,D,O_1\}$ after the game state $(x,y,z)$. Correspondingly, the individual takes action $O_2$ with probability $1-p_{x,y,z}^C-p_{x,y,z}^D-p_{x,y,z}^{O_1}$ after $(x,y,z)$.

We can calculate the long-term payoffs of each individual by a Markov chain. The states of the Markov chain are all possible action profiles, $(A_1,\cdots,A_n)$, where $A_i\in\{C,D,O_1,O_2\}$ is the action taken by individual $i$. Suppose individual $i$ adopts a strategy $\mathbf{p}^{(i)}$. The probability that the game state moves from $(A_1,\cdots,A_n)$ to $(A'_1,\cdots,A'_n)$ is
\begin{equation*}
    l_{(A_1,\cdots,A_n)\rightarrow(A'_1,\cdots,A'_n)}=\prod_{i=1}^n (p_{x,y,z}^{A'_i})^{(i)},
\end{equation*}
where $x$, $y$, and $z$ are, respectively, the number of individuals who take action $C$, $D$, and $O_1$ in the action profile $(A_1,\cdots,A_n)$. Let $\mathbf{L}$ be the transition matrix with entries as defined in the above equation. Under the ``trembling hands'' assumption, $\mathbf{L}$ is primitive and thus has a unique left eigenvector $\mathbf{w}$ corresponding to eigenvalue one. The entry $w_{A_1,\cdots,A_n}$ of $\mathbf{w}$ represents the fraction of time that individuals are found in $(A_1,\cdots,A_n)$. Let $\Pi^{(i)}=(\pi^{(i)}_{A_1,\cdots,A_n})$ denote the payoff vector of individual $i$, where
\begin{equation*}
    \pi^{(i)}_{A_i,\cdots,A_n}=\begin{cases}
        a_{x,x+y},  & A_i=C \text{ and } x+y>1 \\
        b_{x,x+y},    & A_i=D \text{ and } x+y>1 \\
        \sigma,             & otherwise
    \end{cases}.
\end{equation*}
The long-term payoff of individual $i$ is
\begin{equation*}
    \pi^{(i)}=\sum_{A_1,\cdots,A_n}w_{A_1,\cdots,A_n}\pi^{(i)}_{A_1,\cdots,A_n}.
\end{equation*}

Given the long-term payoffs, we can calculate the expected payoff that each individual obtains in randomly formed groups according to the strategy configuration of the population. Furthermore, we simulate the evolution of strategies by stochastic evolutionary dynamics. We find that repeated PGGs with two kinds of actions for opting out yield a higher level of cooperation than repeated PGGs with one action (i.e., repeated optional PGGs described in section \ref{sec:gamesetup}), especially when the multiplication factor is small (Supplementary Fig.~\ref{Sfig:twoadditionaloptions}). This result confirms our conjecture: repeated PGGs with more options to opt out give rise to a higher level of cooperation.
Repeated PGGs with two options to opt out significantly lower the level of defection if the multiplication factor is $r=1.2$ and $r=1.3$. Correspondingly, the level of cooperation is improved markedly for $r=1.2$ and $r=1.3$.

When the multiplication factor is greater than $1.8$, all repeated PGGs support full cooperation. The level of cooperation supported by repeated PGGs with more options to opt out is slightly lower. The reason might be that the strategies that promote the evolution of cooperation in repeated PGGs with more options to opt out are more sensitive to errors. For example, $OD$-$TFT$ and $DO$-$OFT$ need two rounds to rebuild full cooperation if someone defects in State $0$ while $AoN$ only needs one round. Therefore, there is an optimal multiplication factor for repeated PGGs with two kinds of actions for opting out to promote the evolution of cooperation.

\subsection{Feasibility of cooperation in one-shot optional PGGs}
\label{sec:connection_to_oneshot}
In the main text, we show that one-shot optional PGGs hardly support the evolution of cooperation for one specific set of parameters. To demonstrate that one-shot optional PGGs cannot support a high level of cooperation ($>90\%$) and also highlight the importance of repeated interactions in promoting the evolution of cooperation under optional games, here we analyze the feasibility of cooperation in one-shot optional PGGs in details. Specifically, we focus on finitely well-mixed populations, which is consistent with our setting for repeated optional PGGs.

In a one-shot optional PGG, individuals make decisions once and they do so independently. An individual is a cooperator if it plays $C$, a defector if it plays $D$, and a loner if it plays $O$. On the other hand, if we restrict the set of strategies in repeated optional games as $\{ALLC, ALLD, ALLO\}$, it is easy to see that from the perspective of payoffs, playing $C$ ($D$, $O$) in a one-shot PGG is equivalent to adopting $ALLC$ ($ALLD$, $ALLO$) in a repeated optional PGG. This equivalence connects our framework for repeated optional games with that for one-shot optional games. To be compatible with our framework, in the following, we will focus on the evolutionary dynamics among $ALLC$, $ALLD$, and $ALLO$ in repeated optional PGGs, and use them to analyze the feasibility of cooperation in one-shot optional PGGs.

We begin with calculating the average payoffs in the population where there are only two types of strategies. In a population with $j$ individuals of type $A$ and $N-j$ of type $B$, the probability to select $k$ individuals of type $A$ and $n-k$ of type $B$ is
\begin{equation}
    H(k,n,j,N)=\frac{\binom{j}{k}\binom{N-j}{n-k}}{\binom{N}{n}}.
\end{equation}
Thus, in a population with $N_C$ individuals adopting $ALLC$ and $N-N_C$ adopting $ALLD$, the expected payoff of an individual with $ALLC$ is
\begin{equation}
    \pi_{ALLC,ALLD}=\sum_{k=0}^{n-1}H(k,n-1,N_C-1,N-1)\left(\frac{(k+1)rc}{n}-c\right)=\frac{rc}{n}\left(\frac{(N_C-1)(n-1)}{N-1}+1\right)-c,
    \label{equ:payoffallcalld}
\end{equation}
and
the expected payoff of an individual with $ALLD$ is
\begin{equation}
    \pi_{ALLD,ALLC}=\sum_{k=0}^{n-1}H(k,n-1,N_C,N-1)\left(\frac{krc}{n}\right)=\frac{rcN_C(n-1)}{n(N-1)}.
\end{equation}
Similarly, in a population where $N_O$ individuals adopt $ALLO$ and $N-N_O$ adopt $ALLC$, the expected payoff of the individual adopting $ALLO$ is
\begin{equation}
    \pi_{ALLO,ALLC}=\sigma,
\end{equation}
and the expected payoff of an individual with $ALLC$ is
\begin{equation}
    \begin{aligned}
        \pi_{ALLC,ALLO}=&\sum_{k=0}^{n-2}H(k,n-1,N_O,N-1)(r-1)c+\sigma H(n-1,n-1,N_O,N-1)\\
        =&(r-1)c+\frac{\binom{N_O}{n-1}}{\binom{N-1}{n-1}}\left((r-1)c-\sigma\right).
    \end{aligned}
\end{equation}
In a population where $N_O$ individuals adopt $ALLO$ and $N-N_O$ adopt $ALLD$, the expected payoff of the individual adopting $ALLO$ is
\begin{equation}
    \pi_{ALLO,ALLD}=\sigma,
\end{equation}
and the expected payoff of an individual with $ALLD$ is
\begin{equation}
    \pi_{ALLD,ALLO}=\sum_{k=0}^{n-2}H(k,n-1,N_O,N-1)\cdot 0+\sigma H(n-1,n-1,N_O,N-1)=\frac{\binom{N_O}{n-1}}{\binom{N-1}{n-1}}\sigma.
    \label{equ:payoffalldallo}
\end{equation}

We consider the case where mutations are rare ($\mu\rightarrow 0$) such that the population is homogeneous most of the time. There are three different homogeneous states: a population full of individuals using $ALLC$ ($ALLC$ state), using $ALLD$ ($ALLD$ state) and using $ALLO$ ($ALLO$ state). The population transits between these homogeneous states due to occasional mutations. The probability that the population transits from one homogeneous state $\mathbf{p}$ to another state $\mathbf{q}$ depends on the fixation probability that a mutant strategy $\mathbf{p}$ takes over a population full of strategy $\mathbf{q}$, as described in Section \ref{sec:evolutionarydyanmics}. Substituting Eqs.~\eqref{equ:payoffallcalld}-\eqref{equ:payoffalldallo} to Eq.~\eqref{equ:fixation_probability}, it yields
\begin{equation}
    \begin{aligned}
        \rho_{ALLC\rightarrow ALLD}&=\frac{1-\exp\left[-s\left(\frac{rc(n-1)}{n(N-1)}-\frac{rc}{n}+c\right)\right]}{1-\exp\left[-sn\left(\frac{rc(n-1)}{n(N-1)}-\frac{rc}{n}+c\right)\right]},\\
        \rho_{ALLD\rightarrow ALLC}&=\frac{1-\exp\left[-s\left(\frac{rc}{n}-\frac{rc(n-1)}{n(N-1)}-c\right)\right]}{1-\exp\left[-sn\left(\frac{rc}{n}-\frac{rc(n-1)}{n(N-1)}-c\right)\right]},\\
        \rho_{ALLC\rightarrow ALLO}&=\left[1+\sum_{i=1}^{N-1}\exp\left(\frac{s[(r-1)c-\sigma]}{\binom{N-1}{n-1}}\sum_{j=1}^i\binom{j}{n-1}-si[\sigma-(r-1)c]\right)\right]^{-1},\\
        \rho_{ALLO\rightarrow ALLC}&=\left[1+\sum_{i=1}^{N-1}\exp\left(si[\sigma-(r-1)c]-\frac{s[(r-1)c-\sigma]}{\binom{N-1}{n-1}}\sum_{j=1}^i\binom{j}{n-1}\right)\right]^{-1},\\
        \rho_{ALLD\rightarrow ALLO}&=\left[1+\sum_{i=1}^{N-1}\exp\left(s\sigma i-\frac{s\sigma}{\binom{N-1}{n-1}}\sum_{j=1}^i\binom{j}{n-1}\right)\right]^{-1},\\
        \rho_{ALLO\rightarrow ALLD}&=\left[1+\sum_{i=1}^{N-1}\exp\left(\frac{s\sigma}{\binom{N-1}{n-1}}\sum_{j=1}^i\binom{j}{n-1}-s\sigma i\right)\right]^{-1}.\\
    \end{aligned}
\end{equation}
Using $\sum_{j=n-1}^i\binom{j}{n-1}=\sum_{j=n-1}^i\left[\binom{j-1}{n-1}+\binom{j-1}{n-2}\right]$, the last four of the above equations can be simplified to
\begin{equation}
    \begin{aligned}
        \rho_{ALLC\rightarrow ALLO}=&\left[1+\frac{\exp \left(s[(r-1)c-\sigma]\right)-\exp \left(s(n-1)[(r-1)c-\sigma]\right)}{1-\exp \left(s[(r-1)c-\sigma]\right)}\right.\\
        &\left.+\sum_{i=n-1}^{N-1}\exp\left(s[(r-1)c-\sigma]\left(\frac{i+1}{n} \prod_{k=1}^{n-1} \frac{i+1-k}{N-k}+i\right)\right)\right]^{-1},\\
        \rho_{ALLO\rightarrow ALLC}=&\left[1+\frac{\exp \left(-s[(r-1)c-\sigma]\right)-\exp \left(-s(n-1)[(r-1)c-\sigma]\right)}{1-\exp \left(-s[(r-1)c-\sigma]\right)}\right.\\
        &\left.+\sum_{i=n-1}^{N-1}\exp\left(-s[(r-1)c-\sigma]\left(\frac{i+1}{n} \prod_{k=1}^{n-1} \frac{i+1-k}{N-k}+i\right)\right)\right]^{-1},\\
        \rho_{ALLD\rightarrow ALLO}=&\left[1+\frac{\exp \left(s\sigma\right)-\exp \left(s(n-1)\sigma\right)}{1-\exp \left(s\sigma\right)}\right.\\
        &\left.+\sum_{i=n-1}^{N-1}\exp\left(-s\sigma\left(\frac{i+1}{n} \prod_{k=1}^{n-1} \frac{i+1-k}{N-k}+i\right)\right)\right]^{-1},\\
        \rho_{ALLO\rightarrow ALLD}=&\left[1+\frac{\exp \left(-s\sigma\right)-\exp \left(-s(n-1)\sigma\right)}{1-\exp \left(-s\sigma\right)}\right.\\
        &\left.+\sum_{i=n-1}^{N-1}\exp\left(s\sigma\left(\frac{i+1}{n} \prod_{k=1}^{n-1} \frac{i+1-k}{N-k}+i\right)\right)\right]^{-1}.
    \end{aligned}
\end{equation}

Collecting all the fixation probabilities, we have the following transition matrix
\begin{equation}
    \bordermatrix{
	  & ALLC & ALLD & ALLO \cr
	ALLC & 1-\displaystyle\sum_{\mathbf{p}\in\{ALLD,ALLO\}}\frac{\mu}{2}\rho_{ALLC\rightarrow \mathbf{p}} & \frac{\mu}{2}\rho_{ALLC\rightarrow ALLD} &\frac{\mu}{2}\rho_{ALLC\rightarrow ALLO} \cr
	ALLD & \frac{\mu}{2}\rho_{ALLD\rightarrow ALLC} & 1-\displaystyle\sum_{\mathbf{p}\in\{ALLC,ALLO\}}\frac{\mu}{2}\rho_{ALLD\rightarrow \mathbf{p}} & \frac{\mu}{2}\rho_{ALLD\rightarrow ALLO} \cr
    ALLO & \frac{\mu}{2}\rho_{ALLO\rightarrow ALLC} & \frac{\mu}{2}\rho_{ALLO\rightarrow ALLD} & 1-\displaystyle\sum_{\mathbf{p}\in\{ALLC,ALLD\}}\frac{\mu}{2}\rho_{ALLO\rightarrow \mathbf{p}} \cr
	}.
    \label{equ:oneshotoptionaltransition}
\end{equation}
The normalized left eigenvector of the above matrix corresponding to eigenvalue one indicates the probability that the population is in one of the homogeneous states. It is given by
\begin{equation}
    \mathbf{v}=\frac{(\gamma_1,\gamma_2,\gamma_3)}{\gamma_1+\gamma_2+\gamma_3},
    \label{equ:oneshotStationaryDistribution}
\end{equation}
where
\begin{equation*}
    \begin{aligned}
        & \gamma_1=\rho_{ALLD\rightarrow ALLC} \rho_{ALLO\rightarrow ALLC}+\rho_{ALLD\rightarrow ALLC} \rho_{ALLO\rightarrow ALLD}+\rho_{ALLO\rightarrow ALLC} \rho_{ALLD\rightarrow ALLO}, \\
        & \gamma_2=\rho_{ALLO\rightarrow ALLC} \rho_{ALLC\rightarrow ALLD}+\rho_{ALLC\rightarrow ALLD} \rho_{ALLO\rightarrow ALLD}+\rho_{ALLO\rightarrow ALLD} \rho_{ALLC\rightarrow ALLO}, \\
        & \gamma_3=\rho_{ALLD\rightarrow ALLC} \rho_{ALLC\rightarrow ALLO}+\rho_{ALLC\rightarrow ALLD} \rho_{ALLD\rightarrow ALLO}+\rho_{ALLC\rightarrow ALLO} \rho_{ALLD\rightarrow ALLO} .
    \end{aligned}
\end{equation*}

Using the parameters in Fig. 2 of the main text, we calculate the transition matrix as defined in \eqref{equ:oneshotoptionaltransition} and its normalized left eigenvector (Supplementary Fig.~\ref{Sfig:norepeat}\textbf{a}). It can be seen that one-shot optional PGGs hardly support the evolution of cooperation. The population full of $ALLC$ is taken over by a mutant $ALLD$ with a probability of almost one. Subsequently, the population full of $ALLD$ is readily invaded by a mutant $ALLO$. The population full of $ALLO$ is invaded by a mutant $ALLC$ with a probability greater than $1/N$. The evolutionary dynamics exhibit oscillation among full cooperation, full defection and full opt-out.

To test the feasibility of cooperation in one-shot optional PGGs for other model settings, we change the multiplication factor, $r$, and the payoff for opt-out, $\sigma$ (Supplementary Fig.~\ref{Sfig:norepeat}\textbf{b}). It can be seen that one-shot optional PGGs do not support a high level of cooperation ($>90\%$), even if the multiplication factor $r$ reaches its upper limit $n=3$.

\section{Appendix}\label{sec:appendix}
\paragraph*{Proof of Theorem \ref{thm:cooperation}}
($\Leftarrow$) To prove that a strategy $\mathbf{p}$ is a subgame perfect equilibrium (SPE), we make use of the one-shot deviation principle~\cite{Blackwell1965}. We only need to check that there is no profitable one-shot deviation for any individual when all group members use strategy $\mathbf{p}$. In other words, when all group members use strategy $\mathbf{p}$, any individual who deviates in one round and uses $\mathbf{p}$ for all subsequent rounds must obtain a payoff no greater than that it would gain by using $\mathbf{p}$ all the time. By Lemma \ref{lemma:finite_complexity}, we need to check all one-shot deviations in State $0$ (that is, the game states in the previous round are in $\mathcal{H}_C$), State $1$ (the game states in the previous round are in $\mathcal{H}_D$) and State $2$ (the game states in the previous round are in $\mathcal{H}_O$). Here, we take the strategies in the sixth row of Supplementary Fig.~\ref{Sfig:feasiblerangecooperation}\textbf{a} (namely, $OD$-$TFT$) as an example to show how to prove this. For other strategies in Supplementary Fig.~\ref{Sfig:feasiblerangecooperation}\textbf{a}, we provide the long-term payoff obtained by not deviating and by deviating for one round in Supplementary Table \ref{tab:proofthm1}.

Suppose that all individuals adopt $OD$-$TFT$ and individual $i$ considers to deviate. When all individuals are in State $0$,
\begin{itemize}
    \item [(1)] if individual $i$ sticks to $OD$-$TFT$ (i.e., it does not deviate and cooperates in the current round), its long-term payoff is
    \begin{equation}
        \pi_{0C}=(r-1)c;
        \label{equ:pi0c}
    \end{equation}
    \item [(2)] if individual $i$ deviates by defecting in the current round and adheres to $OD$-$TFT$ in all subsequent rounds, its long-term payoff is
    \begin{equation}
        \pi^{(\text{Dev})}_{0D}=(1-\delta)\left[\frac{(n-1)rc}{n}+\delta^2\sigma\right]+\delta^3(r-1)c.
    \end{equation}
    \item [(3)] if individual $i$ deviates by opting out in the current round and using $OD$-$TFT$ for all subsequent rounds, its long-term payoff is at most
    \begin{equation}
        \pi^{(\text{Dev})}_{0O}=(1-\delta)\sigma+\delta(r-1)c.
    \end{equation}
\end{itemize}

When all individuals are in State $1$ and individual $i$ does not deviate, we have
    \begin{equation}
        \pi_{1D}=(1-\delta)\delta\sigma+\delta^2(r-1)c.
    \end{equation}
If individual $i$ deviates by cooperating, its long-term payoff is at most
    \begin{equation}
        \pi^{(\text{Dev})}_{1C}=(1-\delta)\left[(\frac{rc}{n}-c)+\delta\sigma\right]+\delta^2(r-1)c.
    \end{equation}
If individual $i$ deviates by opting out, its long-term payoff becomes
    \begin{equation}
        \pi^{(\text{Dev})}_{1O}=(1-\delta)(1+\delta^2)\sigma+\delta^3(r-1)c.
    \end{equation}

Similarly, when all individuals are in State $2$, if no deviation happens, we have
    \begin{equation}
        \pi_{2O}=(1-\delta)\sigma+\delta(r-1)c.
    \end{equation}
    If individual $i$ deviates, its long-term payoff becomes
    \begin{equation}
        \pi^{(\text{Dev})}_{2C}=(1-\delta)\sigma+\delta(r-1)c,
    \end{equation}
    or
    \begin{equation}
        \pi^{(\text{Dev})}_{2D}=(1-\delta)\sigma+\delta(r-1)c.
        \label{equ:pi2D}
    \end{equation}

Note that $OD$-$TFT$ is a subgame perfect equilibrium if and only if $\pi_{0C}\geq\pi^{(\text{Dev})}_{0D}$, $\pi_{0C}\geq\pi^{(\text{Dev})}_{0O}$, $\pi_{1D}\geq\pi^{(\text{Dev})}_{1C}$, $\pi_{1D}\geq\pi^{(\text{Dev})}_{1O}$, $\pi_{2O}\geq\pi^{(\text{Dev})}_{2C}$ and $\pi_{2O}\geq\pi^{(\text{Dev})}_{2D}$. Substituting the payoffs \eqref{equ:pi0c}-\eqref{equ:pi2D} into these inequalities and taking the limit $\delta\rightarrow1$, we obtain
\begin{equation}
    r\geq\frac{(\sigma+3c)n}{(2n+1)c},
\end{equation}
which is the condition for $OD$-$TFT$ to become a subgame perfect equilibrium in the sixth row of Supplementary Fig.~\ref{Sfig:feasiblerangecooperation}\textbf{a}.

($\Rightarrow$) Because strategy $\mathbf{p}$ is a subgame perfect equilibrium, there exists no profitable one-shot deviation. In particular, we consider three scenarios where all individuals are in State 0, State 1, and State 2.

Let us first consider the case that all individuals are in State $0$ in round $t$. In this case, individual $i$ cooperates and obtains a payoff $(r-1)c$ if it sticks to strategy $\mathbf{p}$. If individual $i$ instead defects in the current round, it gets a payoff $(n-1)rc/n$, which is higher than that under full cooperation. Since individual $i$ has no incentive to deviate in State $0$, the long-term payoff obtained from round $t+1$ on should be less than that yielded by sticking to strategy $\mathbf{p}$. Note that individuals with the same reactive strategy move to the same state. When all individuals are in the same state, full cooperation yields a higher payoff than full opt-out, and full opt-out yields a higher payoff than full defection. Thus, the requirement that there exists no profitable one-shot deviation in State $0$ leads to
\begin{equation}
    p_{n-1,1}^C=0,
    \label{equ:thm1first}
\end{equation}
which means that if one individual defects while everyone else cooperates, strategy $\mathbf{p}$ should not prescribe to cooperate in the next round.

Let us then consider the case that all individuals are in State $1$ in round $t$. In this case, individual $i$ defects and gets a payoff of zero if it adheres to strategy $\mathbf{p}$. If individual $i$ instead opts out, it gets a payoff $\sigma$, which is greater than zero. Similarly, the requirement that individual $i$ has no incentive to deviate in State $1$ means that the long-term payoff obtained from round $t+1$ on should be less than that yielded by sticking to strategy $\mathbf{p}$. Thus, we have
\begin{equation}
    p_{0,n-1}^C=0 \text{ if } p_{0,n}^C=1,
\end{equation}
and
\begin{equation}
    p_{0,n-1}^D=1 \text{ if } p_{0,n}^O=1.
    \label{equ:thm1p0n1}
\end{equation}
If $p_{0,n}^D=1$, the long-term payoff that the one-shot deviation yields from round $t+1$ on is at least the same as that yielded by sticking to strategy $\mathbf{p}$. Thus, the one-shot deviation yields a higher long-term payoff, which is contradictory to the fact that strategy $\mathbf{p}$ is a subgame perfect equilibrium. Thus
\begin{equation}
    p_{0,n}^D=0.
    \label{equ:thm1p0n}
\end{equation}

Let us now consider the case that all individuals are in State $2$ in round $t$. Individual $i$ gains the same payoff $\sigma$ in round $t$ for both sticking to strategy $\mathbf{p}$ and deviating. Whether the one-shot deviation is profitable depends on the states visited afterward. Note that individuals stay in State $0$ once all individuals move to State $0$ (guaranteed by $p^C_{n,0}=1$). Note also that full cooperation yields the highest possible payoff. In this case, the strategy that leads to an earlier transition into State $0$ yields a higher payoff. Therefore, we have
\begin{equation}
    p_{1,0}^C=p_{0,1}^C=0 \text{ if } p_{0,0}^D=p_{0,n}^C=1
\end{equation}
On the other hand, if $p_{0,n}^O=1$, individuals alternate between full defection and full opt-out. The requirement that individual $i$ has no incentive to deviate in State $1$ implies that individuals can not move to State $0$ by deviating in State $1$. Combining with Eq. \eqref{equ:thm1p0n1}, we have
\begin{equation}
    p_{1,n-1}^C=0 \text{ and } p_{0,n-1}^D=1 \text{ if } p_{0,0}^D=p_{0,n}^O=1.
\end{equation}
In addition, in the case where $p_{0,0}^D=p_{0,n}^O=p_{1,0}^O=1$, individual $i$ gets $\sigma/(1+\delta)$ if it sticks to strategy $\mathbf{p}$ and it obtains a higher payoff $(1-\delta)\sigma+\delta \sigma/(1+\delta)$  if it cooperates for one round and adheres to strategy $\mathbf{p}$ in all subsequent rounds. The case where $p_{0,0}^D=p_{0,n}^O=p_{0,1}^O=1$ leads to a similar result, which is contradictory to the fact that individual $i$ has no incentive to deviate in State $2$. Thus, we have
\begin{equation}
    p_{1,0}^D=p_{0,1}^D=1 \text{ if } p_{0,0}^D=p_{0,n}^O=1.
\end{equation}
If $p_{0,0}^O=1$, individuals stay in State $2$ for all rounds. We also need $p_{0,1}^C=p_{1,0}^C=0$, namely,
\begin{equation}
    p_{1,0}^C=p_{0,1}^C=0 \text{ if } p_{0,0}^O=1.
\end{equation}
Furthermore, if individuals move to State $1$ due to a one-shot deviation ($p_{0,1}^D=1$ or $p_{1,0}^D=1$), the requirement that individuals cannot move to State $0$ earlier implies that game states $(0,n)$, $(1,n-1)$ and $(0,n-1)$ do not lead to State $0$. Combining with Eq. \eqref{equ:thm1p0n1} and \eqref{equ:thm1p0n}, the conditions
\begin{equation}
    p_{1,n-1}^C=0 \text{ and } p_{0,n}^O=p_{0,n-1}^D=1 \text{ if } p_{0,0}^O=p_{1,0}^D=1
\end{equation}
and
\begin{equation}
    p_{1,n-1}^C=0 \text{ and } p_{0,n}^O=p_{0,n-1}^D=1 \text{ if } p_{0,0}^O=p_{0,1}^D=1
    \label{equ:thm1last}
\end{equation}
are necessary.

Combining Eqs.~\eqref{equ:thm1first}-\eqref{equ:thm1last}, we characterize all strategies that support persistent cooperation and are equilibria in Supplementary Fig.~\ref{Sfig:feasiblerangecooperation}\textbf{a}.

$\hfill\qedsymbol$

\paragraph*{Proof of Theorem \ref{thm:defection}}
To prove the theorem, suppose on the contrary that there is a subgame perfect equilibrium $\mathbf{p}$ that supports defection. A strategy can favor defection only if it defects after all individuals defect. That is, $p_{0,n}^D=1$ is necessary. If all individuals are in State $1$, every individual gets zero by sticking to $\mathbf{p}$. If an individual opts out for one round, it at least obtains $(1-\delta)\sigma$. Thus, there is an incentive to make a one-shot deviation from $\mathbf{p}$.

\paragraph*{Proof of Theorem \ref{thm:optout}}
($\Leftarrow$) Similar to the proof of Theorem \ref{thm:cooperation}, we make use of the one-shot deviation principle~\cite{Blackwell1965} to prove that a strategy $\mathbf{p}$ is a subgame perfect equilibrium. We only need to check that there is no profitable one-shot deviation when all group members use strategy $\mathbf{p}$. By Lemma \ref{lemma:finite_complexity}, we need to check the cases that individuals are in State $0$ (that is, the game state in the previous round belongs to $\mathcal{H}_C$), State $1$ (that is, the game state in the previous round belongs to $\mathcal{H}_D$), and State $2$ (that is, the game state in the previous round belongs to $\mathcal{H}_O$). For each strategy in Supplementary Fig.~\ref{Sfig:feasiblerangeoptout}\textbf{a}, we list the long-term payoffs that un-deviation and one-shot deviation yield in Supplementary Table \ref{tab:proofthmoptout}. Let $\pi_{0C}$, $\pi_{1D}$ and $\pi_{2O}$ denote individual $i$'s long-term payoff when all individuals are in State $0$, State $1$ and State $2$, and individual $i$ sticks to its strategy, respectively. Let $\pi^{(\text{Dev})}_{sA}$ denote individual $i$'s long-term payoff when all individuals are in State $s\in\{0,1,2\}$ and individual $i$ deviates by taking action $A\in\{C,D,O\}$ for one round and acts as others in all subsequent rounds. The strategy is a subgame equilibrium if and only if $\pi_{0C}\geq\pi^{(\text{Dev})}_{0D}$, $\pi_{0C}\geq\pi^{(\text{Dev})}_{0O}$, $\pi_{1D}\geq\pi^{(\text{Dev})}_{1C}$, $\pi_{1D}\geq\pi^{(\text{Dev})}_{1O}$, $\pi_{2O}\geq\pi^{(\text{Dev})}_{2C}$ and $\pi_{2O}\geq\pi^{(\text{Dev})}_{2D}$. Solving these inequalities, we get the conditions listed in Supplementary Fig.~\ref{Sfig:feasiblerangeoptout}\textbf{a}.

($\Rightarrow$) Since $\mathbf{p}$ is a subgame perfect equilibrium, individual $i$ with strategy $\mathbf{p}$ has no incentive to make one-shot deviations. Let us first consider the case that all individuals are in State $0$ in round $t$. Similar to the proof of Theorem \ref{thm:cooperation}, if individual $i$ makes a one-shot deviation to obtain a higher payoff in the current round, the requirement that individual $i$ has no incentive to deviate in State $0$ implies that the long-term payoff that the deviation yields from round $t+1$ on is lower. Note that individuals with the same pure reactive strategy move to the same state and full cooperation yields more benefit than full opt-out, full opt-out yields more benefit than full defection. Thus, we have
\begin{equation}
    p_{n-1,1}^C=0 \text{ if } p_{n,0}^C=1,
    \label{equ:thm3first}
\end{equation}
and
\begin{equation}
    p_{n-1,1}^D=1 \text{ if } p_{n,0}^O=1.
\end{equation}
If $p_{n,0}^D=1$, the one-shot deviation brings at least the same payoff as un-deviation after round $t$. Thus, the total long-term payoff that individual $i$ gets by defecting for one round is higher than that yielded by sticking to strategy $\mathbf{p}$, which is contradictory to the fact that strategy $\mathbf{p}$ is a subgame perfect equilibrium. Therefore,
\begin{equation}
    p_{n,0}^D=0.
\end{equation}

Similarly, when all individuals are in State $1$ in round $t$, individual $i$ can get a higher payoff in the current round if it opts out. The requirement that individual $i$ has no incentive to deviate in State $1$ implies that the long-term payoff that the one-shot deviation yields from round $t+1$ on must be lower. This results in
\begin{equation}
    p_{0,n-1}^C=0 \text{ if } p_{0,n}^C=1,
    \label{equ:thm30n11}
\end{equation}
and
\begin{equation}
    p_{0,n-1}^D=1 \text{ if } p_{0,n}^O=1.
    \label{equ:thm30n1}
\end{equation}
If $p_{0,n}^D=1$, the one-shot deviation is always profitable, which is contradictory to the fact that strategy $\mathbf{p}$ is a subgame perfect equilibrium. Therefore,
\begin{equation}
    p_{0,n}^D=0.
\end{equation}

Consider the case that all individuals are in State $2$ in round $t$. If individual $i$ sticks to strategy $\mathbf{p}$, it obtains $\sigma$ in the current round. If individual $i$ instead cooperates (defects) in the current round, it gets $\sigma$ in this round, which is not lower than that obtained by sticking to strategy $\mathbf{p}$. If the deviation makes individuals move to State $0$ (i.e. $p_{1,0}^C=1$ or $p_{0,1}^C=1$), individual $i$ gets $\delta(r-1)c$ in the subsequent rounds in the case where $p_{n,0}^C=1$ and gets $(1-\delta)\delta(r-1)c+\delta^2\sigma$ in the case where $p_{n,0}^O=1$. Both the payoffs are greater than the long-term payoff obtained by sticking to strategy $\mathbf{p}$ from round $t+1$ on, $\delta\sigma$. Thus, there is an incentive for individual $i$ to deviation in State $2$, which contradicts the fact that strategy $\mathbf{p}$ is a subgame perfect equilibrium. Thus, the condition
\begin{equation}
    p_{1,0}^C=p_{0,1}^C=0
\end{equation}
is necessary.

Let us now consider some more complex cases.
\begin{itemize}
    \item [(1)] The case where $p_{0,n}^O=p_{n,0}^C=p_{1,n-1}^C=1$. In this case, if all individuals are in State $1$, individual $i$ gets $\delta\sigma$ by sticking to strategy $\mathbf{p}$. Instead, individual $i$ can get $(1-\delta)(\frac{rc}{n}-c)+\delta(r-1)c$ by cooperating for one round and then sticking to strategy $\mathbf{p}$. Thus, such a deviation is profitable. To prevent this deviation, $p_{1,n-1}^C=0$ is required. This means
    \begin{equation}
        p_{1,n-1}^C=0 \text{ if } p_{0,n}^O=p_{n,0}^C=1.
    \end{equation}
    \item [(2)] The case where $p_{0,n}^C=p_{n,0}^C=1$ and $p_{0,1}^D=1$ ($p_{1,0}^D=1$). Similar to the above case, when all individuals are in State $2$, the long-term payoff of individual $i$ is $\sigma$ if individual $i$ sticks to strategy $\mathbf{p}$ while its long-term payoff is $(r-1)c$ if it deviates in State $2$. The one-shot deviation yields a higher long-term payoff, which is contradictory to the fact that individual $i$ has no incentive to deviate in State $2$. Thus, we have
    \begin{equation}
        p_{0,1}^O=p_{1,0}^O=1 \text{ if } p_{0,n}^C=p_{n,0}^C=1.
        \label{equ:thm3last}
    \end{equation}
    \item [(3)] The case where $p_{n,0}^O=p_{n-1,1}^D=p_{0,n}^C=p_{0,n-1}^O=1$. In this case, if all individuals are in State $0$, individual $i$ gets $\pi_{0C}=(1-\delta)(r-1)c+\delta\sigma$ by adhering to strategy $\mathbf{p}$ and gets $\pi^{(\text{Dev})}_{0D}=(1-\delta)(\frac{(n-1)rc}{n}+\delta^2(r-1)c)+\delta^3\sigma$ by defecting for one round and sticking to strategy $\mathbf{p}$ in all subsequent rounds. The requirement that individual $i$ has no incentive to deviate in State $0$ implies $\pi_{0C}\geq\pi^{(\text{Dev})}_{0D}$. Taking the limit $\delta\rightarrow1$, we get
    \begin{equation}
        \sigma\geq\frac{(n-1)rc}{2n}.
        \label{equ:thm3payoff1}
    \end{equation}
    If all individuals are in State $1$, individual $i$ obtains $\pi_{1D}=(1-\delta)\delta(r-1)c+\delta^2\sigma$ by adhering to strategy $\mathbf{p}$ while obtains $\pi^{(\text{Dev})}_{1O}=\sigma$ by opting out for one round and using strategy $\mathbf{p}$ in all subsequent rounds. The requirement that individual $i$ has no incentive to deviate in State $1$ implies $\pi_{1D}\geq\pi^{(\text{Dev})}_{1O}$. Taking the limit $\delta\rightarrow1$, we obtain
    \begin{equation}
        \sigma\leq\frac{(r-1)c}{2}.
        \label{equ:thm3payoff2}
    \end{equation}
    Eqs. \eqref{equ:thm3payoff1} and \eqref{equ:thm3payoff2} can not be satisfied simultaneously in repeated optional PGGs. Thus, strategy $\mathbf{p}$ with $p_{n,0}^O=p_{n-1,1}^D=p_{0,n}^C=p_{0,n-1}^O=1$ is not a subgame perfect equilibrium. To exclude this case, we get
    \begin{equation*}
        p_{0,n-1}^O=0 \text{ if }p_{n,0}^O=p_{n-1,1}^D=p_{0,n}^C=1.
    \end{equation*}
    Combining the above equation with Eq. \eqref{equ:thm30n11}, we get
    \begin{equation}
        p_{0,n-1}^D=1 \text{ if }p_{n,0}^O=p_{n-1,1}^D=p_{0,n}^C=1.
        \label{equ:thm3lllast}
    \end{equation}
\end{itemize}

Combining Eqs.~\eqref{equ:thm3first}-\eqref{equ:thm3last} and Eq.~\eqref{equ:thm3lllast}, we get Supplementary Fig.~\ref{Sfig:feasiblerangeoptout}\textbf{a}.
$\hfill\qedsymbol$

\setcounter{section}{0}
\renewcommand{\figurename}{Supplementary Fig.}
\renewcommand{\tablename}{Supplementary Table}

\clearpage
\textbf{\large Supplementary Figures}
\vspace*{\fill}
\begin{center}
    \begin{figure}[H]
        \centering
        \includegraphics[width=1\textwidth]{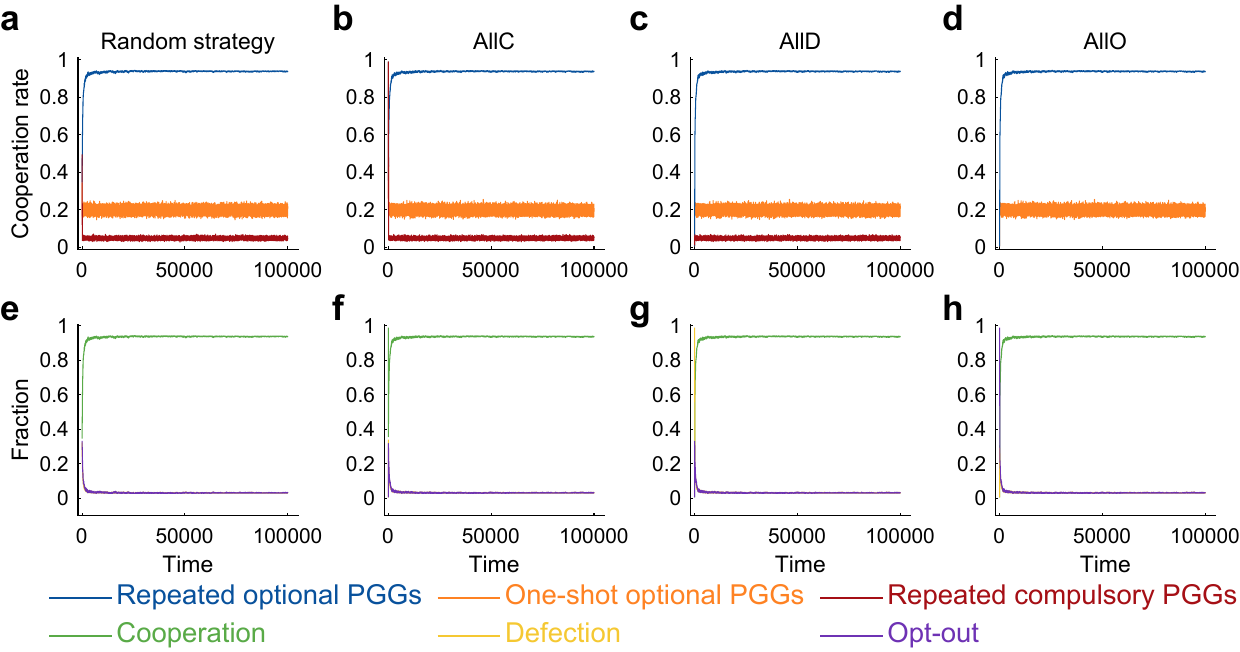}
        \caption[Robustness to initial strategy configurations]{\textbf{Our findings are robust to the initial strategy configurations of the population.} To investigate the robustness of our findings to the initial strategy configurations of the population, we explore the evolution of cooperation in the population with four of the most important configurations: random strategy (\textbf{a}, \textbf{e}), $ALLC$ (\textbf{b}, \textbf{f}), $ALLD$ (\textbf{c}, \textbf{g}), and $ALLO$ (\textbf{d}, \textbf{h}). Repeated optional PGGs support full cooperation while repeated compulsory PGGs and one-shot optional PGGs hardly promote the evolution of cooperation, regardless of the initial configurations (\textbf{a}-\textbf{d}). The additional option, opt-out, acts as a catalyst for the emergence of cooperation especially when the strategy configuration of the population is initialized as $ALLD$ (\textbf{g}). On the other hand, opt-out prevents the further expansion of defection when $ALLC$ occupies the population (\textbf{f}). The parameters are the same as in Fig. 2 of the main text.}
        \label{Sfig:2random}
    \end{figure}
\end{center}
\vspace{\fill}

\clearpage
\vspace*{\fill}
\begin{center}
    \begin{figure}[H]
        \centering
        \includegraphics[width=\textwidth]{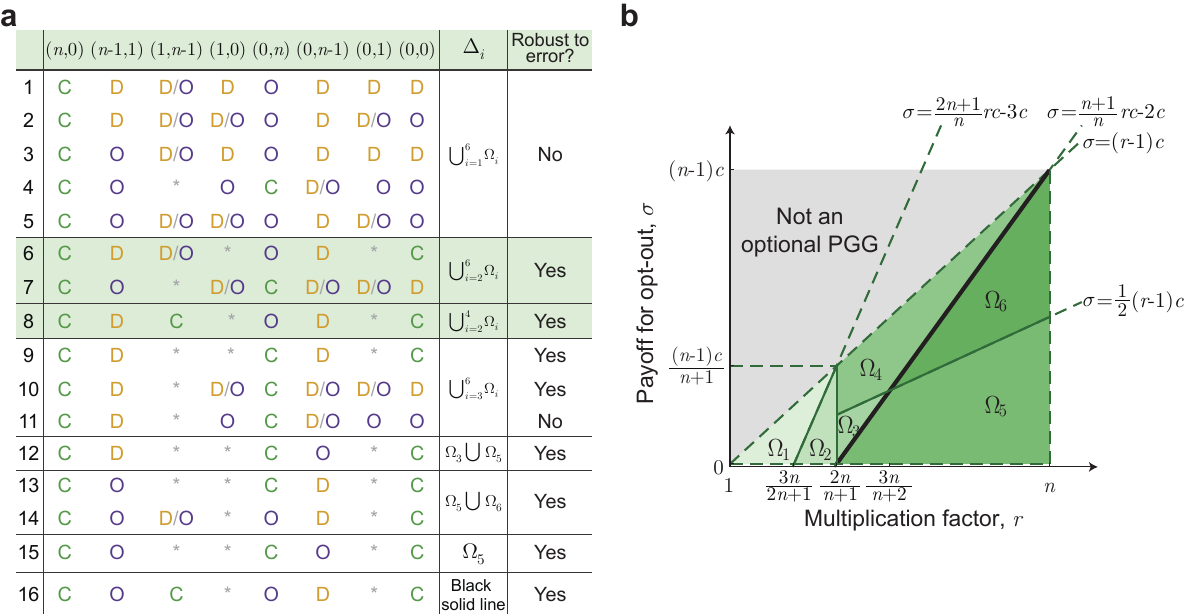}
        \caption[Pure reactive strategies that support persistent full cooperation and form an equilibrium]{\textbf{$OD$-$TFT$, $DO$-$OFT$, and $OD$-$TFT_C$ have the lowest threshold for the multiplication factor to surpass to become equilibria within strategies that robustly support cooperation.} \textbf{a}, We identify all subgame perfect equilibria (SPE) that support persistent cooperation in the space of pure reactive strategies. These strategies cooperate in the very first round (i.e. $p_{0}^C=1$) and stick to cooperation if all other individuals do so (i.e. $p_{n,0}^C=1$). We classify these strategies according to the actions taken after each game state $(x,y)$ (the second to ninth columns) and the feasible region $\Delta_i$ to become a subgame perfect equilibrium (SPE). In repeated optional PGGs, the strategies in the first to fifth rows are always SPE. Yet they are prone to errors: if all individuals adopt one of these strategies and are affected by implementation errors, cooperation breaks down. The strategies in the eleventh row are also not robust to errors.  \textbf{b}, We draw the feasible region that each strategy becomes an SPE. As shown, among the strategies that robustly support cooperation, $OD$-$TFT$ (strategies in the sixth row), $DO$-$OFT$ (strategies in the seventh row), and $OD$-$TFT_C$ (strategies in the eighth row) become SPE with the lowest threshold for the multiplication factor to surpass. Here, symbol $D$/$O$ means that strategies can defect or opt out after the corresponding game state. The symbol ``$*$'' represents that all available actions, $C$, $D$, and $O$, are allowed after the corresponding game state. Note that for the game states that are not listed, all actions are allowed.
        }
        \label{Sfig:feasiblerangecooperation}
    \end{figure}
\end{center}
\vspace{\fill}

\clearpage
\vspace*{\fill}
\begin{center}
    \begin{figure}[H]
        \centering
        \includegraphics[width=0.9\textwidth]{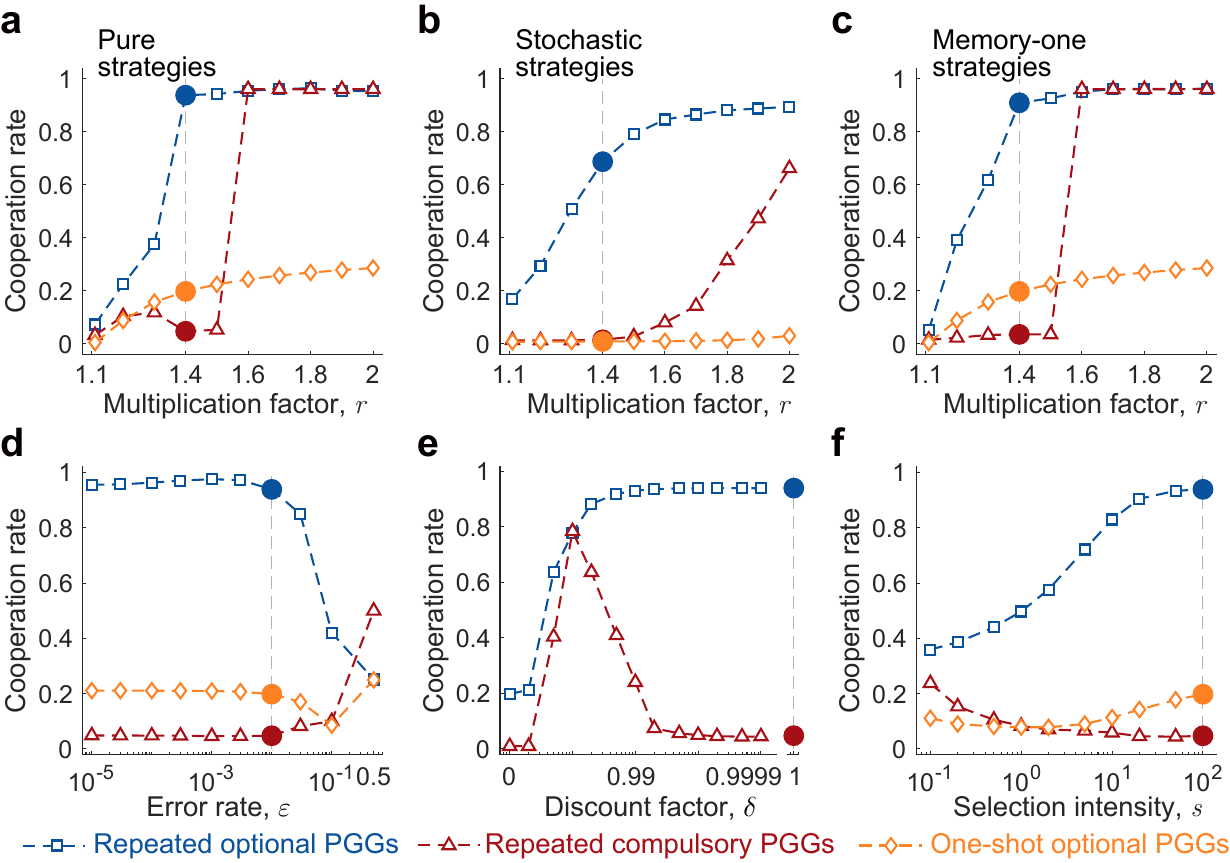}
        \caption[Robustness to parameter changes]{\textbf{Our findings are robust to parameter changes.} To investigate the robustness of our findings to the changes of different parameters, we independently vary the key parameters in Fig. 2 of the main text. We find that the additional option, opt-out, has a positive effect on the evolution of cooperation over a wide range of the multiplication factor $r$ (\textbf{a}, \textbf{b}, \textbf{c}), the error rate $\varepsilon$ (\textbf{d}), the discount factor $\delta$ (\textbf{e}), or the selection intensity $s$ (\textbf{f}). In all cases, repeated optional PGGs yield a cooperation premium until the error rate exceeds some threshold. Here, each data point is the average of $1000$ independent simulations. Filled circles mark the parameter settings in Fig. 2 of the main text. The default parameter values are $N=100$, $c=1$, $r=1.4$, $\sigma=0.1$, $s=100$, $\varepsilon=0.01$ and $\delta\rightarrow1$.}
        \label{Sfig:influence}
    \end{figure}
\end{center}
\vspace{\fill}

\clearpage
\vspace*{\fill}
\begin{center}
    \begin{figure}[H]
        \centering
        \includegraphics[width=0.8\textwidth]{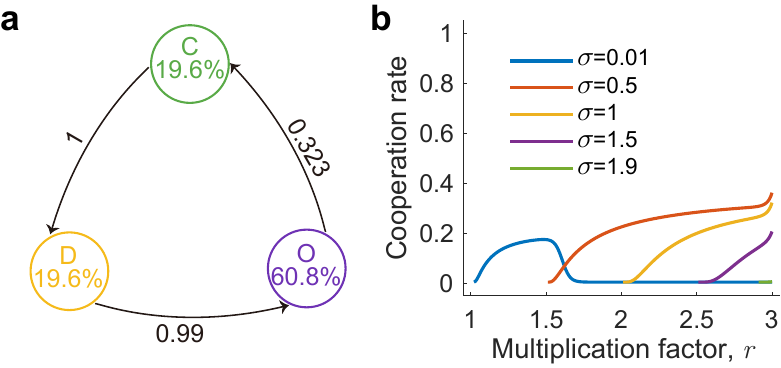}
        \caption[Rock-paper-scissors type of cycling in one-shot optional PGGs]{\textbf{One-shot optional PGGs lead to a rock-paper-scissors type of cycling among cooperation, defection and opt-out.} In a one-shot optional PGG, if an individual is a cooperator (defector, nonparticipant), this individual is assumed to play $C$ ($D$, $O$) independent of other individuals' previous decisions. \textbf{a}, We analyze the evolutionary dynamics of one-shot optional PGGs in finitely well-mixed populations. The numbers within the circles are the average abundance of cooperators, defectors and nonparticipants. The numbers near the arrows represent the fixation probability that the mutant at the end of the arrows takes over the population composed entirely of individuals who adopt the strategy of the starting point of the arrow. It can be seen that a mutant cooperator readily invades a population of nonparticipants; a mutant defector readily invades a population of cooperators; and a mutant nonparticipant readily invades a population of defectors. It exhibits a dynamic oscillation among cooperation, defection and nonparticipation. Thus, one-shot optional PGGs hardly support full cooperation. \textbf{b}, According to Eq.~\eqref{equ:oneshotStationaryDistribution}, we further calculate the cooperation rate for various values of multiplication factors, $r$, and payoffs for opt-out, $\sigma$. It can be seen that one-shot optional PGGs hardly support full cooperation (the maximum cooperation rate is less than $40\%$) although the multiplication factor approaches the group size. The parameter values in panel \textbf{a} and the parameter values except $r$ and $\sigma$ in panel \textbf{b} are the same as those in Fig. 2 of the main text.}
        \label{Sfig:norepeat}
    \end{figure}
\end{center}
\vspace{\fill}

\clearpage
\vspace*{\fill}
\begin{center}
    \begin{figure}[H]
        \centering
        \includegraphics[width=0.9\textwidth]{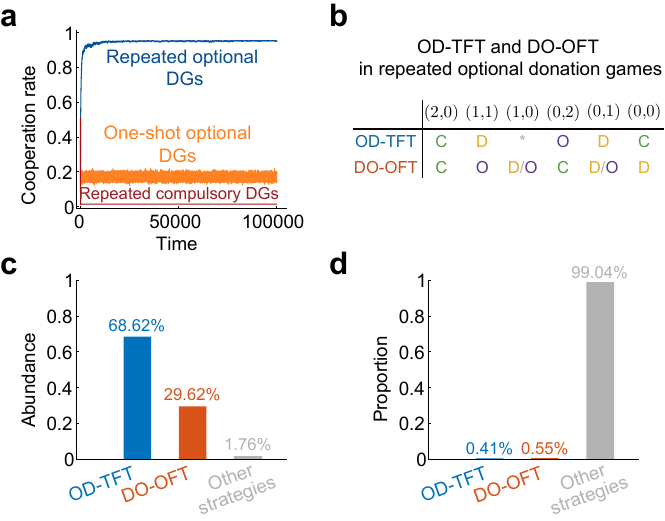}
        \caption[Evolution of cooperation in repeated optional donation games]{\textbf{Voluntary participation facilitates  cooperation in repeated donation games (DGs).} \textbf{a}, We explore the effect of opt-out on the evolution of cooperation in classical repeated pairwise games, i.e., the repeated DGs. In each round of repeated optional DGs, individuals can decide to cooperate by donating a benefit $b'$ to its opponent at a cost $c'$, to defect by denoting nothing and bearing no cost, or to opt out and obtain $\sigma$. Similar to repeated optional PGGs, the game is canceled if only one individual participates. In this case, each individual gets $\sigma$. As shown, repeated optional DGs support the evolution of cooperation even if the corresponding one-shot optional DGs and repeated compulsory DGs fail to. \textbf{b}, According to Theorem \ref{thm:cooperation}, we get the representations of $OD$-$TFT$ and $DO$-$OFT$ in repeated optional DGs. \textbf{c,d}, $OD$-$TFT$ and $DO$-$OFT$ are also the dominant strategies in repeated optional DGs even if they occupy a very low proportion in the space of pure reactive strategies. Parameters: $b=1.65$, $\sigma=0.1$. Other parameters are the same as in Fig. 2 of the main text.}
        \label{Sfig:donationGame}
    \end{figure}
\end{center}
\vspace{\fill}

\clearpage
\vspace*{\fill}
\begin{center}
    \begin{figure}[!ht]
        \centering
        \includegraphics[width=0.8\textwidth]{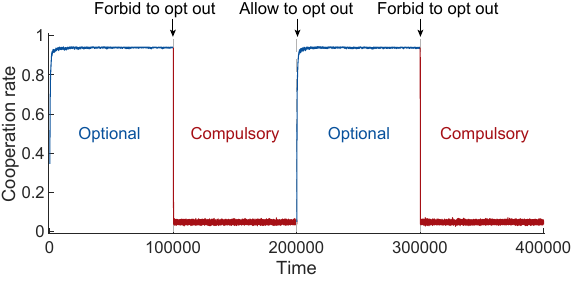}
        \caption[The switch-like role that opt-out plays in repeated games]{\textbf{The additional option, opt-out, acts as a switch for the evolution of cooperation.} To better illustrate the role that opt-out plays in supporting cooperation, we explore how the cooperation rate varies as the game transits between repeated optional and compulsory PGGs. Here, we assume that the transitions occur at Time 100000, 200000, and 300000. It turns out that if individuals are allowed to opt out of the interaction (from Time 0 to Time 100000), the population is in an almost fully cooperative state; once individuals are deprived of the option to opt out, the population rapidly falls into the deadlock of full defection (from Time 100001 to Time 200000); if the right to opt out is given again to all the individuals, the almost fully cooperative state is quickly recovered (from Time 200001 to Time 300000), and so on. The parameter values are the same as those in Fig. 2 of the main text.}
        \label{Sfig:transition}
    \end{figure}
\end{center}
\vspace{\fill}

\clearpage
\vspace*{\fill}
\begin{center}
    \begin{figure}[H]
        \centering
        \includegraphics[width=1\textwidth]{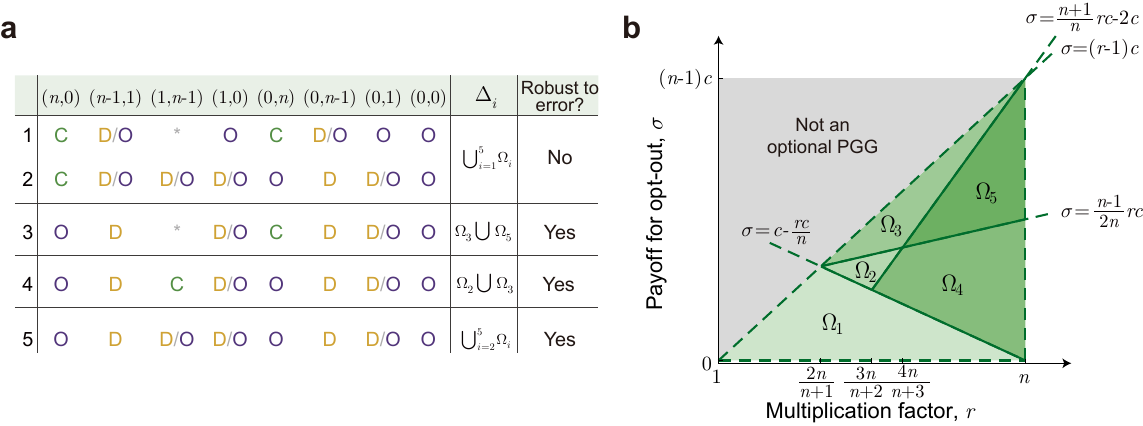}
        \caption[Pure reactive strategies that support persistent full opt-out and form an equilibrium]{\textbf{Five pure reactive strategies give rise to persistent full opt-out and form an equilibrium.} \textbf{a}, We identify all pure reactive strategies that support persistent full opt-out and form an equilibrium. Each strategy opts out in the very first round (i.e. $p_0^O=1$) and sticks to opt-out once all individuals do so (i.e. $p_{0,0}^O=1$). All strategies except for the strategy in the first and second rows are robust to errors. \textbf{b}, Feasible regions of strategies listed in panel \textbf{a} to become equilibria are illustrated. As shown, the strategy in the fifth row in panel \textbf{a} has the largest feasible region to robustly support persistent full opt-out. The meanings of symbols are the same as those in Supplementary Fig.~\ref{Sfig:feasiblerangecooperation}.
        }
        \label{Sfig:feasiblerangeoptout}
    \end{figure}
\end{center}
\vspace{\fill}

\clearpage
\vspace*{\fill}
\begin{center}
    \begin{figure}[H]
        \centering
        \includegraphics[width=0.9\textwidth]{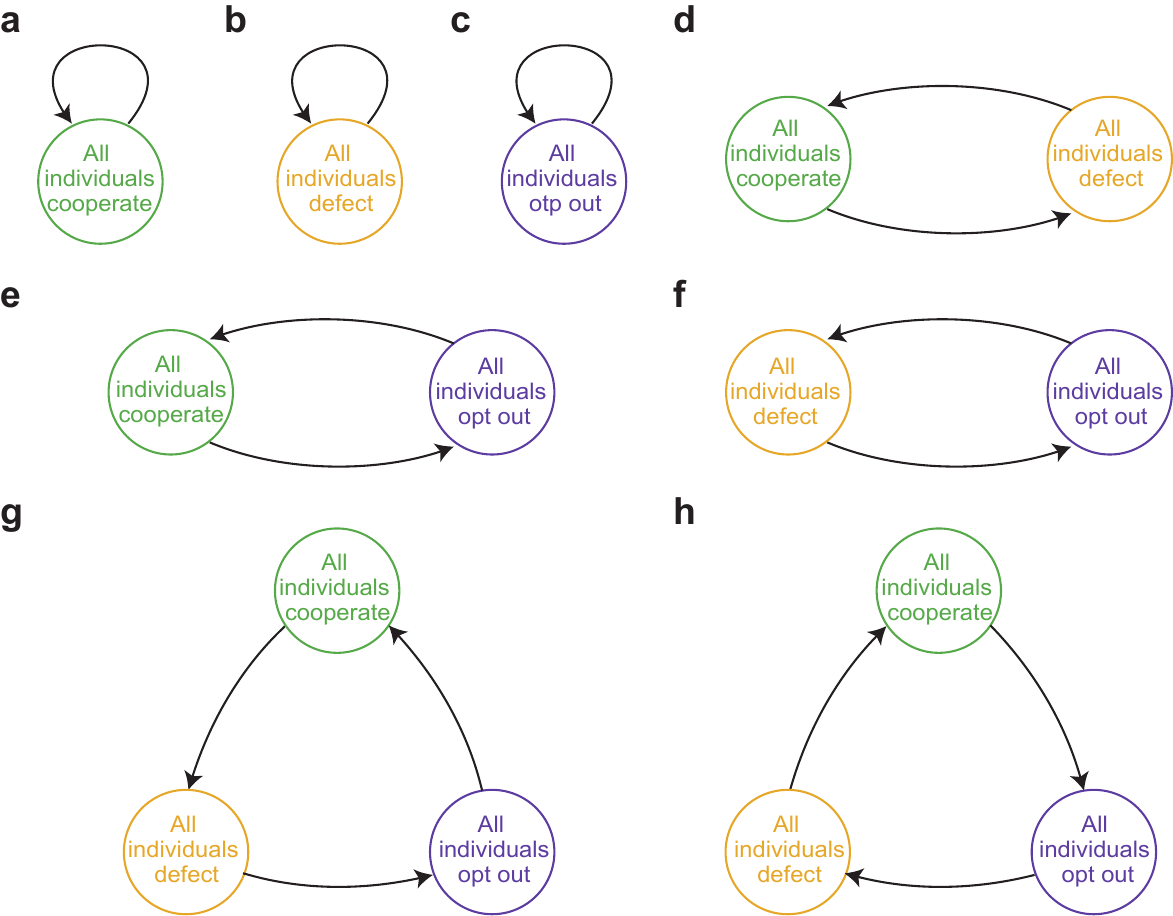}
        \caption[Eight possible endings of game dynamics]{\textbf{A group of individuals that adopt the same pure reactive strategy leads to at most eight possible endings of game dynamics in repeated optional PGGs.} The node corresponds to the game state of the group in a certain round. There are three possible game states: all individuals cooperate, all defect and all opt out. The group moves from one state in one round to another state in the next round along the direction of the arrow. The eight possible endings of game dynamics are all possible transitions between these states. Panels \textbf{a}-\textbf{h} correspond to endings \emph{\textbf{1}}-\emph{\textbf{8}} in Theorem \ref{lemma:dynamicsofgroup}.}
        \label{Sfig:8dynamics}
    \end{figure}
\end{center}
\vspace{\fill}

\clearpage
\vspace*{\fill}
\begin{center}
    \begin{figure}[H]
        \centering
        \includegraphics[width=0.9\textwidth]{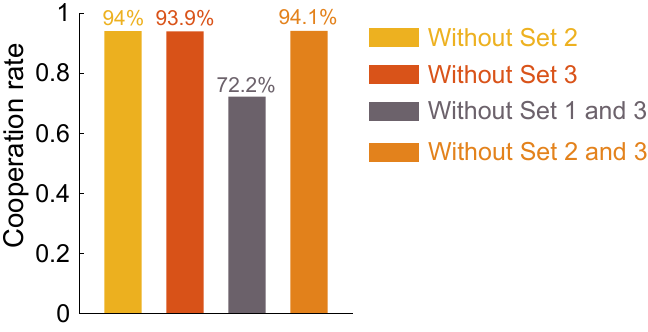}
        \caption[Supplementary information for the ``knock-out'' experiments]{\textbf{$OD$-$TFT$, $DO$-$OFT$ and $OD$-$TFT_C$ exert the preeminent impact on the evolution of cooperation.} To further investigate the effect of $OD$-$TFT$, $DO$-$OFT$, and their behaviorally close variants, we supplement the ``knock-out'' experiments by deleting other combinations of these strategies apart from the combinations shown in Fig.~2 in the main text. Again, we denote the set consisting of $OD$-$TFT$, $DO$-$OFT$, and $OD$-$TFT_C$ as Set 1, the one-bit variants of $OD$-$TFT$ and $DO$-$OFT$ as Set 2 and the two-bit variants as Set 3. The results of ``knock-out'' experiments indicate that repeated optional games support full cooperation if strategies in Set 1 are available, independently of the availability of strategies in Set 2 and 3. If strategies in Set 1 are unavailable but those in Set 2 except $OD$-$TFT_C$ are accessible, the cooperation rate is about $70\%$, regardless of whether strategies in Set 3 are allowed or not. Thus, $OD$-$TFT$, $DO$-$OFT$, and $OD$-$TFT_C$ exert the leading impact on the evolution of cooperation, succeeded by one-bit variants of $OD$-$TFT$ and $DO$-$OFT$, and ultimately by the two-bit variants. The parameters are the same as those in Fig.~2 of the main text.
        }
        \label{Sfig:KnockOutExperimentSupplement}
    \end{figure}
\end{center}
\vspace{\fill}

\clearpage
\vspace*{\fill}
\begin{center}
    \begin{figure}[H]
        \centering
        \includegraphics[width=0.7\textwidth]{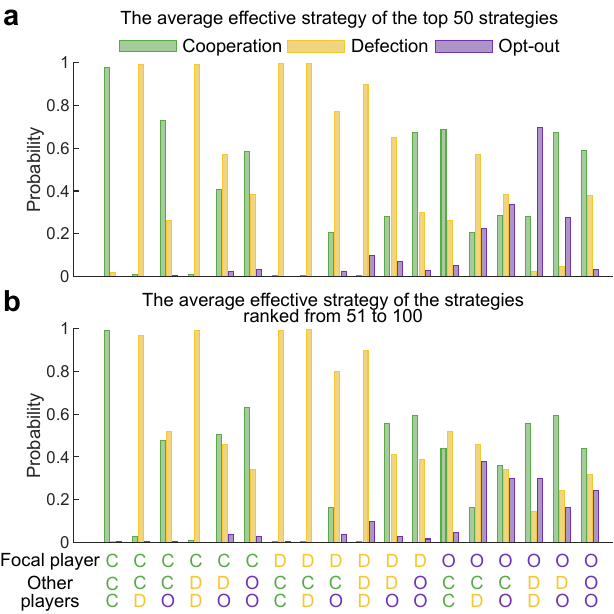}
        \caption[The average effective strategy of the top 100 strategies in repeated optional PGGs with failures of opting out]{\textbf{The dominant strategies ensure that individuals restore cooperation quickly when failures of opting out exist.} To better illustrate the underlying mechanism that repeated optional PGGs with failures promote cooperation, we calculate the average effective strategy of the top 50 strategies (\textbf{a}) and that of the strategies ranked from 51 to 100 (\textbf{b}) in the evolutionary dynamics. These effective strategies behave similarly to $OD$-$TFT$. Besides, they exhibit additional properties to ensure that individuals rebuild cooperation quickly. When one individual successfully opts out after game state $(D,DD)$, most strategies of the top 50 strategies prescribe opt-out after game states $(D,DO)$ and $(O,DD)$, and prescribe cooperation when at least two individuals opt out successfully. In this case, most strategies ranked between 51 and 100 prescribe cooperation after game states $(D,DO)$ and $(O,DD)$, directly restoring cooperation. To obtain the average effective strategies, we perform 1000 independent simulations. In each simulation, $10^5$ mutants are introduced to the population. The failure probability $\alpha$ is $0.9$ and other parameters are the same as those in Fig.~2 of the main text.}
        \label{Sfig:delayed}
    \end{figure}
\end{center}
\vspace{\fill}

\clearpage
\vspace*{\fill}
\begin{center}
    \begin{figure}[H]
        \centering
        \includegraphics[width=1\textwidth]{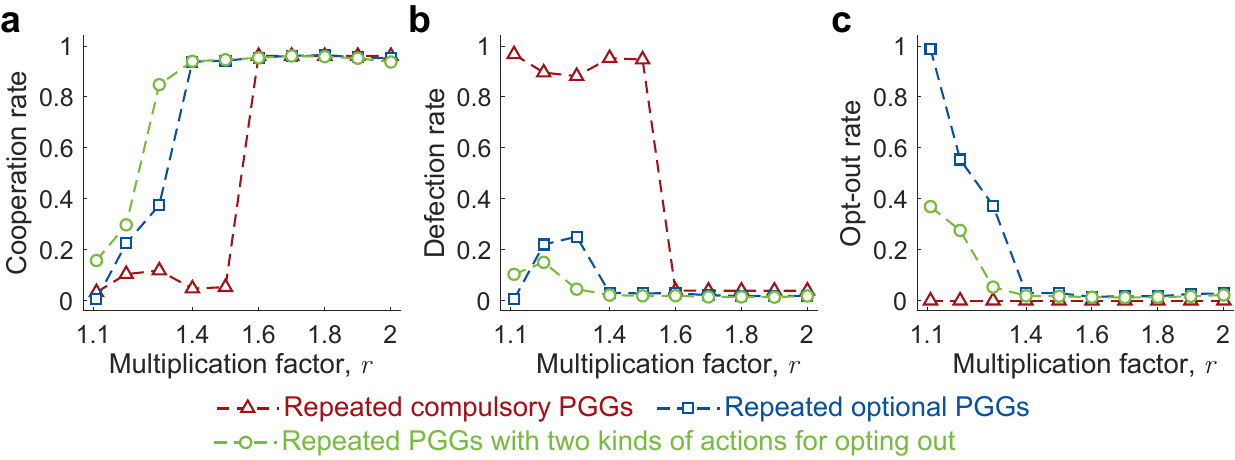}
        \caption[Higher level of cooperation in repeated PGGs with two actions for opting out]{\textbf{Repeated PGGs with two opt-out options further improve the level of cooperation.} Theoretical analysis indicates that reactive strategies $OD$-$TFT$ and $DO$-$OFT$ have similar properties as memory-two strategy $AoN_2$ has. This inspires us to explore whether repeated PGGs with more options can further improve the level of cooperation. To test this conjecture, we introduce repeated PGGs with two kinds of actions for opting out, $O_1$ and $O_2$. An individual who takes action $O_1$ or $O_2$ opts out of the interaction and gets a fixed payoff of $\sigma$. We run $4000$ simulations and compute the average cooperation rate (\textbf{a}), defection rate (\textbf{b}), and opt-out rate (\textbf{c}) over the $4000$ simulations under different multiplication factors. It can be seen that repeated PGGs with more opt-out options give rise to a higher level of cooperation, especially if the multiplication factor is small. In panel \textbf{c}, the opt-out rate of repeated PGGs with two kinds of actions for opting out is the sum of fractions of $O_1$ and $O_2$. The parameters are the same as in Fig. 2 of the main text.}
        \label{Sfig:twoadditionaloptions}
    \end{figure}
\end{center}
\vspace{\fill}

\setcounter{table}{0}
\definecolor{lightgreen}{RGB}{168,210,159}
\definecolor{lightorange}{RGB}{239,211,128}
\definecolor{lightpurple}{RGB}{173,148,193}

\clearpage
\renewcommand{\arraystretch}{2.2}
\setlength\tabcolsep{2pt}
\begin{landscape}
    \begin{table}
        \centering
        \tiny
        \begin{tabular}{c||c|c|c||c|c|c||c|c|c}
            & \multicolumn{3}{c||}{\scriptsize All individuals are in State $0$} & \multicolumn{3}{c||}{\scriptsize All individuals are in State $1$} & \multicolumn{3}{c}{\scriptsize All individuals are in State $2$} \\
            \hline
            & \makecell{Sticking to\\cooperation,\\$\pi_{0C}$} & \makecell{Defect for\\one round,\\$\pi^{(\text{Dev})}_{0D}$} & \makecell{Opt out for\\one round,\\$\pi^{(\text{Dev})}_{0O}$} & \makecell{Cooperate for\\one round,\\$\pi^{(\text{Dev})}_{1C}$} & \makecell{Stick to\\defection,\\$\pi_{1D}$} & \makecell{Opt out for\\one round,\\$\pi^{(\text{Dev})}_{1O}$} & \makecell{Cooperate for\\one round,\\$\pi^{(\text{Dev})}_{2C}$} & \makecell{Defect for\\one round,\\$\pi^{(\text{Dev})}_{2D}$} & \makecell{Stick to\\opt-out,\\$\pi_{2O}$} \\
            \hline
            1 & \multirow{5}{*}{$(r-1)c$} & $(1-\delta)\frac{(n-1)rc}{n}+\frac{\delta^2\sigma}{1+\delta}$ & \multirow{5}{*}{\makecell{$(1-\delta)\sigma$\\+$\delta(r-1)c$}}
            & $(1-\delta)(\frac{rc}{n}-c)+\frac{\delta\sigma}{1+\delta}$ & $\frac{\delta\sigma}{1+\delta}$ & $\frac{\sigma}{1+\delta}$
            & $\frac{\sigma}{1+\delta}$ & $\frac{\sigma}{1+\delta}$ & $\frac{\sigma}{1+\delta}$ \\
            2 & & $(1-\delta)\frac{(n-1)rc}{n}+\delta^2\sigma$ &
            & $(1-\delta)(\frac{rc}{c}-c)+\delta\sigma$ & $\delta\sigma$ & $(1-\delta+\delta^2)\sigma$
            & $\sigma$ & $\sigma$ & $\sigma$ \\
            3 & & $(1-\delta)\frac{(n-1)rc}{n}+\frac{\delta\sigma}{1+\delta}$ &
            & $(1-\delta)(\frac{rc}{n}-c)+\frac{\delta\sigma}{1+\delta}$ & $\frac{\delta\sigma}{1+\delta}$ & $\frac{\sigma}{1+\delta}$
            & $\frac{\sigma}{1+\delta}$ & $\frac{\sigma}{1+\delta}$ & $\frac{\sigma}{1+\delta}$ \\
            4 & & $(1-\delta)\frac{(n-1)rc}{n}+\delta\sigma$ &
            & \makecell{$(1-\delta)(\frac{rc}{n}-c)$\\+$\delta(r-1)c$} & $\delta(r-1)c$ & $(1-\delta)\sigma+\delta^2(r-1)c$
            & $\sigma$ & $\sigma$ & $\sigma$ \\
            5 & & $(1-\delta)\frac{(n-1)rc}{n}+\delta\sigma$ &
            & $(1-\delta)(\frac{rc}{n}-c)+\delta\sigma$ & $\delta\sigma$ & $(1-\delta+\delta^2)\sigma$
            & $\sigma$ & $\sigma$ & $\sigma$\\
            \hline
            6 & \multirow{2}{*}{$(r-1)c$}  & \makecell{$(1-\delta)\left[\frac{(n-1)rc}{n}+\delta^2\sigma\right]$\\+$\delta^3(r-1)c$} & \multirow{2}{*}{\makecell{$(1-\delta)\sigma$\\+$\delta(r-1)c$}}
            & \makecell{$(1-\delta)\left[(\frac{rc}{n}-c)+\delta\sigma\right]$\\+$\delta^2(r-1)c$} & $(1-\delta)\delta\sigma+\delta^2(r-1)c$ & \makecell{$(1-\delta)(1+\delta^2)\sigma$\\+$\delta^3(r-1)c$}
            & $(1-\delta)\sigma+\delta(r-1)c$ & $(1-\delta)\sigma+\delta(r-1)c$ &$(1-\delta)\sigma+\delta(r-1)c$ \\
            7 & & \makecell{$(1-\delta)\left[\frac{(n-1)rc}{n}+\delta\sigma\right]$\\+$\delta^3(r-1)c$} &
            & $(1-\delta)(\frac{rc}{n}-c)+\delta(r-1)c$ & $\delta(r-1)c$ & $(1-\delta)\sigma+\delta^2(r-1)c$
            & $(1-\delta)\sigma+\delta^2(r-1)c$ & $(1-\delta)\sigma+\delta^2(r-1)c$ & $(1-\delta)\sigma+\delta^2(r-1)c$ \\
            \hline
            8 &\multirow{1}{*}{$(r-1)c$} & \makecell{$(1-\delta)\left[\frac{(n-1)rc}{n}+\delta^2\sigma\right]$\\+$\delta^3(r-1)c$} & \makecell{$(1-\delta)\sigma$\\+$\delta(r-1)c$}
            & $(1-\delta)(\frac{rc}{n}-c)+\delta(r-1)c$ & $(1-\delta)\delta\sigma+\delta^2(r-1)c$ & \makecell{$(1-\delta)(1+\delta^2)\sigma$\\$+\delta^3(r-1)c$}
            & $(1-\delta)\sigma+\delta(r-1)c$ & $(1-\delta)\sigma+\delta(r-1)c$ &$(1-\delta)\sigma+\delta(r-1)c$ \\
            \hline
            9 &\multirow{3}{*}{$(r-1)c$}  & \multirow{3}{*}{\makecell{$(1-\delta)\frac{(n-1)rc}{n}$\\+$\delta^2(r-1)c$}} & \multirow{3}{*}{\makecell{$(1-\delta)\sigma$\\+$\delta(r-1)c$}}
            & \multirow{3}{*}{\makecell{$(1-\delta)(\frac{rc}{n}-c)$\\+$\delta(r-1)c$}} & \multirow{3}{*}{$\delta(r-1)c$} & \multirow{3}{*}{$(1-\delta)\sigma+\delta^2(r-1)c$}
            & $(1-\delta)\sigma+\delta(r-1)c$ & $(1-\delta)\sigma+\delta(r-1)c$ & $(1-\delta)\sigma+\delta(r-1)c$ \\
            10 & &  &
            &  & &
            & $(1-\delta)\sigma+\delta^2(r-1)c$ & $(1-\delta)\sigma+\delta^2(r-1)c$ & $(1-\delta)\sigma+\delta^2(r-1)c$ \\
            11 & &  &
            &  &  &
            & $\sigma$ & $\sigma$ & $\sigma$ \\
            \hline
            12 & $(r-1)c$ & \makecell{$(1-\delta)\frac{(n-1)rc}{n}$\\+$\delta^2(r-1)c$} & \makecell{$(1-\delta)\sigma$\\+$\delta(r-1)c$}
            & \makecell{$(1-\delta)(\frac{rc}{n}-c)$\\+$\delta(r-1)c$} & $\delta(r-1)c$ & $(1-\delta^2)\sigma+\delta^2(r-1)c$
            & $(1-\delta)\sigma+\delta(r-1)c$ & $(1-\delta)\sigma+\delta(r-1)c$ & $(1-\delta)\sigma+\delta(r-1)c$ \\
            \hline
            13 &\multirow{2}{*}{$(r-1)c$} & \multirow{2}{*}{\makecell{$(1-\delta)\left[\frac{(n-1)rc}{n}+\delta\sigma\right]$\\+$\delta^2(r-1)c$}} & \multirow{2}{*}{\makecell{$(1-\delta)\sigma$\\+$\delta(r-1)c$}}
            & $(1-\delta)(\frac{rc}{n}-c)+\delta(r-1)c$ & $\delta(r-1)c$ & $(1-\delta)\sigma+\delta^2(r-1)c$
            & \multirow{2}{*}{\makecell{$(1-\delta)\sigma$\\+$\delta(r-1)c$}} & \multirow{2}{*}{\makecell{$(1-\delta)\sigma$\\+$\delta(r-1)c$}} & \multirow{2}{*}{\makecell{$(1-\delta)\sigma$\\+$\delta(r-1)c$}} \\
            14 & & &
            & \makecell{$(1-\delta)(\frac{rc}{n}-c+\delta\sigma)$\\+$\delta^2(r-1)c$} & $(1-\delta)\delta\sigma+\delta^2(r-1)c$ & \makecell{$(1-\delta)(1+\delta^2)\sigma$\\$+\delta^3(r-1)c$}
            & & & \\
            \hline
            15 &$(r-1)c$ & \makecell{$(1-\delta)\left[\frac{(n-1)rc}{n}+\delta\sigma\right]$\\+$\delta^2(r-1)c$} &     \makecell{$(1-\delta)\sigma$\\+$\delta(r-1)c$}
            & \makecell{$(1-\delta)(\frac{rc}{n}-c)$\\+$\delta(r-1)c$} & $\delta(r-1)c$ & \makecell{$(1-\delta^2)\sigma$\\+$\delta^2(r-1)c$}
            & $(1-\delta)\sigma+\delta(r-1)c$ & $(1-\delta)\sigma+\delta(r-1)c$ & \makecell{$(1-\delta)\sigma$\\+$\delta(r-1)c$} \\
            \hline
            16 &$(r-1)c$ & \makecell{$(1-\delta)\left[\frac{(n-1)rc}{n}+\delta\sigma\right]$\\+$\delta^2(r-1)c$} &   \makecell{$(1-\delta)\sigma$\\+$\delta(r-1)c$}
            & \makecell{$(1-\delta)(\frac{rc}{n}-c)$\\+$\delta(r-1)c$} &$(1-\delta)\delta\sigma+\delta^2(r-1)c$ & \makecell{$(1-\delta)(1+\delta^2)\sigma$\\$+\delta^3(r-1)c$}
            &  $(1-\delta)\sigma+\delta(r-1)c$ & $(1-\delta)\sigma+\delta(r-1)c$ & \makecell{$(1-\delta)\sigma$\\+$\delta(r-1)c$} \\
        \end{tabular}
        \caption[Supplementary information for the proof of Theorem \ref{thm:cooperation}]{\textbf{The long-term payoffs obtained by sticking to each strategy in Supplementary Fig.~\ref{Sfig:feasiblerangecooperation}\textbf{a} and by deviating for one round are listed.} This table is the supplementary information for the proof of Theorem \ref{thm:cooperation}. For the strategy with $D/O$ and $*$, the table exhibits the maximal long-term payoff. Each strategy is an equilibrium if and only if $\pi_{0C}\geq\pi^{(\text{Dev})}_{0D}$, $\pi_{0C}\geq\pi^{(\text{Dev})}_{0O}$, $\pi_{1D}\geq\pi^{(\text{Dev})}_{1C}$, $\pi_{1D}\geq\pi^{(\text{Dev})}_{1O}$, $\pi_{2O}\geq\pi^{(\text{Dev})}_{2C}$ and $\pi_{2O}\geq\pi^{(\text{Dev})}_{2D}$. Solving these inequalities and taking limit $\delta\rightarrow1$, we get the feasible range of the strategies in Supplementary Fig.~\ref{Sfig:feasiblerangecooperation}\textbf{a} to become equilibria.}
        \label{tab:proofthm1}
    \end{table}
\end{landscape}

\renewcommand{\arraystretch}{3.5}
\setlength\tabcolsep{2pt}
\begin{landscape}
    \begin{table}
        \centering
        \tiny
        \begin{tabular}{c||c|c|c||c|c|c||c|c|c}
            & \multicolumn{3}{c||}{\scriptsize All individuals are in State $0$} & \multicolumn{3}{c||}{\scriptsize All individuals are in State $1$} & \multicolumn{3}{c}{\scriptsize All individuals are in State $2$} \\
            \hline
            & \makecell{Sticking to\\cooperation,\\$\pi_{0C}$} & \makecell{Defect for\\one round,\\$\pi^{(\text{Dev})}_{0D}$} & \makecell{Opt out for\\one round,\\$\pi^{(\text{Dev})}_{0O}$} & \makecell{Cooperate for\\one round,\\$\pi^{(\text{Dev})}_{1C}$} & \makecell{Stick to\\defection,\\$\pi_{1D}$} & \makecell{Opt out for\\one round,\\$\pi^{(\text{Dev})}_{1O}$} & \makecell{Cooperate for\\one round,\\$\pi^{(\text{Dev})}_{2C}$} & \makecell{Defect for\\one round,\\$\pi^{(\text{Dev})}_{2D}$} & \makecell{Stick to\\opt-out,\\$\pi_{2O}$} \\
            \hline
        1 & \multirow{2}{*}{$(r-1)c$} & $(1-\delta)\frac{(n-1)rc}{n}+\delta^2(r-1)c$ & \multirow{2}{*}{$(1-\delta)\sigma+\delta(r-1)c$}
          & $(1-\delta)(\frac{rc}{n}-c)+\delta(r-1)c$ & $\delta(r-1)c$ & $(1-\delta)\sigma+\delta^2(r-1)c$
          & \multirow{2}{*}{$\sigma$} & \multirow{2}{*}{$\sigma$} & \multirow{2}{*}{$\sigma$} \\
        2 & & $(1-\delta)\frac{(n-1)rc}{n}+\delta\sigma$ &
          & $(1-\delta)(\frac{rc}{n}-c)+\delta\sigma$ & $\delta\sigma$ & $(1-\delta+\delta^2)\sigma$
          & & & \\
        \hline
        3 & $(1-\delta)(r-1)c+\delta\sigma$ & \makecell{$(1-\delta)(\frac{(n-1)rc}{n}$\\$+\delta^2(r-1)c)+\delta^3\sigma$} & $(1-\delta)(\sigma+\delta(r-1)c)+\delta^2\sigma$
          & \makecell{$(1-\delta)((\frac{rc}{n}-c)$\\$+\delta(r-1)c)+\delta^2\sigma$} & $(1-\delta)\delta(r-1)c+\delta^2\sigma$ & \makecell{$(1-\delta)(\sigma$\\$+\delta^2(r-1)c)+\delta^3\sigma$}
          & $\sigma$ & $\sigma$ & $\sigma$ \\
          \hline
        4 & $(1-\delta)(r-1)c+\delta\sigma$ & $(1-\delta)\frac{(n-1)rc}{n}+\delta^2\sigma$ & $(1-\delta)(\sigma+\delta(r-1)c)+\delta^2\sigma$
          & \makecell{$(1-\delta)((\frac{rc}{n}-c)$\\$+\delta(r-1)c)+\delta^2\sigma$} & $\delta\sigma$ & $(1-\delta+\delta^2)\sigma$
          & $\sigma$ & $\sigma$ & $\sigma$ \\
        \hline
        5 & $(1-\delta)(r-1)c+\delta\sigma$ & $(1-\delta)\frac{(n-1)rc}{n}+\delta^2\sigma$ & $(1-\delta)(\sigma+\delta(r-1)c)+\delta^2\sigma$
           & $(1-\delta)(\frac{rc}{n}-c)+\delta\sigma$ & $\delta\sigma$ & $(1-\delta+\delta^2)\sigma$
           & $\sigma$ & $\sigma$ & $\sigma$ \\
        \end{tabular}
        \caption[Supplementary information for the proof of Theorem \ref{thm:optout}]{\textbf{The long-term payoffs obtained by sticking to each strategy in Supplementary Fig.~\ref{Sfig:feasiblerangeoptout}\textbf{a} and by deviating for one round are listed.} This table is the supplementary information for the proof of Theorem \ref{thm:optout}. For the strategy with $D/O$ and $*$, the table exhibits the maximal long-term payoff. Each strategy is an equilibrium if and only if $\pi_{0C}\geq\pi^{(\text{Dev})}_{0D}$, $\pi_{0C}\geq\pi^{(\text{Dev})}_{0O}$, $\pi_{1D}\geq\pi^{(\text{Dev})}_{1C}$, $\pi_{1D}\geq\pi^{(\text{Dev})}_{1O}$, $\pi_{2O}\geq\pi^{(\text{Dev})}_{2C}$ and $\pi_{2O}\geq\pi^{(\text{Dev})}_{2D}$. Solving these inequalities and taking limit $\delta\rightarrow1$, we get the feasible range for strategies in Supplementary Fig.~\ref{Sfig:feasiblerangeoptout}\textbf{a} to become equilibria.}
        \label{tab:proofthmoptout}
    \end{table}
\end{landscape}

\end{document}